\documentclass{article}

\usepackage[british]{babel}
\usepackage{hyphenat}
\usepackage[section]{placeins}
\usepackage{amsmath}
\usepackage{graphicx}
\usepackage{algorithm}
\usepackage{algorithmicx}
\usepackage{algpseudocode}
\usepackage{subcaption}
\usepackage{float}   
\usepackage{caption}
\usepackage{multicol,lipsum}
\usepackage{mathtools}
\usepackage{cuted}
\usepackage{tcolorbox}
\usepackage{amsfonts}
\usepackage{amsthm}
\usepackage{amssymb}
\usepackage{booktabs}

\newtheorem*{claim*}{Claim}
\usepackage[colorinlistoftodos,prependcaption,textsize=small]{todonotes}

\makeatletter

\AtBeginDocument{
  \setlength\abovedisplayskip{0pt}
  \setlength\belowdisplayskip{0pt}}

\newenvironment{stbox}{
  \begin{tcolorbox}[colback=gray!5!white,colframe=gray!75!black]
}{
  \end{tcolorbox}
}
\newtcolorbox[auto counter, number within=subsection]{mybox}[2][]{
    title=\thetcbcounter: #2, label=box:#2,#1
}

\title{Inverse-Designed Dot Product Engine}
\author{Anannya Mathur\\
Department of SIT/CSE, IIT Delhi\\
\texttt{siy237565@iitd.ac.in} \quad \texttt{anannyamathur2000@gmail.com}
}
\date{} 

\usepackage[hidelinks]{hyperref}
\begin{document}

\maketitle

\begin{abstract}
The work presents an inverse-designed optical cavity that can direct light from two sources such that if the sources
were to represent any number in the range [-1,1] with magnitude encoded through the power emitted by the source and sign
by switching the direction of source current, the photocurrent generated at
the two output ports is proportional to the product of the two numbers. Let us say that the two sources encode x and y,
which are two numbers $\in$ [-1,1]. Multiplication is reduced to the form $(x+y)^2 - (x-y)^2 = 4xy \propto xy$. The addition and subtraction operations of the numbers are supported by constructive and destructive interference, respectively. The work shows that replacing the DDOT dot product engine of the Lightening Transformer with the optical cavity proposed to calculate 
the dot product can lead to a reduction in the area occupied by the photonic core by 88 \%, can reduce the power consumption by lasers by around 23.43 \%, and bring down energy consumption while training DeiT models by 0.88 \%. The cavities can generate photocurrents of the form $1.057 xy + 0.249$ with $R^2=0.88,$ thus showing a relationship of direct proportionality between the target product $xy$ and the output of the cavity in response to stimuli encoding $x$ and $y$. 

\end{abstract}

\maketitle

\section{Introduction}

Lightening Transformer  \cite{zhu2023lighteningtransformerdynamicallyoperatedopticallyinterconnectedphotonic}, is the
most recent state-of-the-art work that reports the lowest energy consumption and area amongst all existing
``unconventional'' accelerators. It proposes optical vector dot products. Their system uses Mach-Zehnder modulators as
wavelength demultiplexers, beam splitters and phase shifters. Using their open-source simulation package, the
simulation results for the DeiT-Tiny workload (image transformer) showed an energy consumption of 1.129 mJ while the
standard accelerators reported 36.82 mJ (precision = 8 bits, number of tiles = 8 and number of PEs per tile = 2) for
the same workload. Their
photonic accelerator could reduce energy consumption by 32.6$\times$. Their solution is not area-efficient.   

Although traditional beam splitters and switches are extremely efficient and show great promise, they have large area
footprints \cite{liu2019}. The shift to inverse-designed cavity designs is prompted by the need to develop high-density
devices that can promote the integration of tens of billions of components on a single photonic integrated circuit.
Therefore, optical cavities are a viable replacement for traditional Microring Resonators and MZMs with the
potential to support ultra-compact devices that can compute complex functions. The wave nature of light can be used to support multifunction computation beyond simple bending and wavelength demultiplexing, thus eliminating the need to integrate multiple components in an integrated circuit, which ultimately results in reduced power consumption, reduced chip area, and lower latency. One such application is the inverse designed cavity that accepts a waveform encoded as a mathematical function and outputs the integral of the function
\cite{camacho2021single} \cite{estakhri2019inverse}. They inverse-designed their cavities by constraining their design
space in the cm scale, but inverse designs have the potential to constrain the cavities in the nanometer scale, reducing their footprints in the micrometer scale. One such finding reported are cavity designs to miniaturize particle accelerators, where it was shown that an additional kick of around 0.9 kiloelectron volts (keV) can be given to
a bunch of 80-keV electrons along just 30 micrometers of a specially designed channel \cite{Sapra_2020}. They
implemented a waveguide-integrated dielectric laser accelerator, and therefore reduced the scale of such accelerators by $10^4$ using light-electron interaction. 

To the best of our knowledge, there is no prior work on inverse-designed dot product engines. Given that dot products
are the most basic units of computation in neural networks, exploiting the wave nature of light can provide substantial
gains in terms of power and performance. Through the application of wave geometry and interference, we aim to
demonstrate an optical cavity with dimensions in nanometers. This cavity will efficiently compute the dot product of two
numbers at the speed of light. Consequently, conventional optical components, such as phase shifters, directional
couplers, Mach Zehnder interferometer, and Y-branch, will be unnecessary. The biggest challenge in such work is that the
results (when converted to an electrical domain) can be approximate. This may introduce errors. This paper aims to study
the exact nature of errors and to achieve an equitable trade-off. 



\href{ https://github.com/stanfordnqp/spins-b}{SPINS-B} is an open source simulation package maintained by Stanford
University and can be used to inverse design optical cavities. The package supports inverse designs when only a single
source of light is present in the system (such as the package can design cavities to mimic beam splitters, wavelength
demultiplexers, bending of light, and so on). The paper hopes to lay down the theory and propose an extension to the
SPINS package to solve optical systems where multiple sources are present, and the cavities need to be reconstructed to
steer appropriate light-media interactions of the multiple source-sink pairs (the dot product happens to be a version of
one of these optical systems where two numbers represent two distinct sources of light).

\section{Background}

\subsection{Maxwell's Equations}
Let $\vec{B}$ be the magnetic induction field, $\vec{E}$ be the electric field, 
$\vec{D}$ be the displacement field (also known as electric flux), 
$\vec{H}$ be the magnetic field, $\rho$ be the free charge density (charge per unit volume), 
$\epsilon'(\vec{r},\omega) \left ( = \frac{\epsilon(\vec{r},\omega)}{\epsilon_0} \right )$ be the
relative permittivity or dielectric function where $\vec{r}$ is the position vector and $\omega$ is the frequency. 
$J$ is the current density (current per unit area). Maxwell's equations are as follows.

\begin{equation}
\label{eqn:divb}
    \nabla.\vec{B} =0
\end{equation}

According to Green's theorem,
Equation~\ref{eqn:divb} can be rewritten as $\iint_S(\vec{B}.\hat{n})dS=0$,
where $S$ is a closed
surface, $\hat{n}$ is the unit surface normal and $dS$ is an infinitesimal small piece of the surface.
This means that the net magnetic flux through a closed surface is zero. 

\begin{equation}
    \nabla \times \vec{E} + \frac{\partial \vec{B}}{\partial t}=0
    \label{eqn:ebr}
\end{equation}

Equation~\ref{eqn:ebr} is also known as the Maxwell-Faraday equation~\cite{2023Maxwell}. It can
be rewritten as $\oint_C \vec{E}.dl= -\iint_S (\frac{\partial \vec{B}}{\partial t}.\hat{n})ds$.
$C$ is a closed path, $S$ is the
surface bounded by $C$, $dl$ is an infinitesimally small vector element on the closed path $C$ and
$ds$ is an 
infinitesimally small surface element. The law signifies that the electric potential associated
with a closed path C is entirely due to the time-varying nature of the magnetic field.

\begin{equation}
    \nabla.\vec{D}= \rho
    \label{eqn:drho}
\end{equation}

Equation~\ref{eqn:drho} quantifies the amount of free charge at a point and equates it with the divergence ($\nabla.$) of the electric flux.

\begin{equation}
    \nabla \times \vec{H} = \frac{\partial \vec{D}}{\partial t} + \vec{J} 
    \label{eqn:hdj}
\end{equation}

Equation~\ref{eqn:hdj} signifies that even in the absence of a current source ($\vec{J}$), a changing electric field
($\vec{D}$) can induce a magnetic field. This is the Maxwell-Ampere law, which establishes the theoretical
foundations of electromagnetic (EM) waves.

\subsection{EM Waves}

Let us assume that the refractive index, absorption or transmission of the material remains constant and does not depend on the intensity, amplitude or frequency of the incoming light
wave \cite{photonic_crystal}. Let us consider anisotropic media
where the permittivity varies. In this case $\epsilon$ is a 
$3 \times 3$ matrix. 
The displacement field $\vec{D}$ can thus be represented as ${D}_i = \epsilon_0 \sum_j \epsilon_{ij}' E_j$, where $i$ and $j$ index the $x,y,z$ axes of the coordinate system, $D_i$ is one of the $x,y,z$ components of $\vec{D}$, and $\epsilon_0$ is the electrical permittivity of the free space. 

If we restrict ourselves to macroscopic and isotropic media, then
$D$ becomes

\begin{equation}
    \vec{D}(\vec{r})= \epsilon_o \epsilon'(\vec{r})\vec{E}(\vec{r})=  \epsilon(\vec{r}) \vec{E}(\vec{r})
    \label{eqn:de}
\end{equation}

The magnetic induction field $\vec{B}$ is related to the magnetic field $\vec{H}$ via $\vec{B}(\vec{r})=\mu_0 \mu(\vec{r}) \vec{H}(\vec{r})$. 

\begin{equation}
    \vec{B}(\vec{r})=\mu_0 \mu(\vec{r}) \vec{H}(\vec{r})
    \label{eqn:bhmu}
\end{equation}

Let us consider $\mu(\vec{r})= 1$ \cite{photonic_crystal} where $u(\vec{r})$ is the relative magnetic permeability. Therefore, Equation~\ref{eqn:bhmu} becomes

\begin{equation}
     \vec{B}(\vec{r})=\mu_0 \vec{H}(\vec{r})
     \label{eqn:bh}
\end{equation}

From equations \ref{eqn:divb} and \ref{eqn:bh}, the following relation can be derived.

\begin{equation}
\nabla.H(\vec{r},t)=0
    \label{eqn:divH}
\end{equation}

Equation~\ref{eqn:divH} shows that the magnetic fields are perpendicular to wave propagation. It is also known as the Tranversality Rule of Magnetic Fields. 

Using Equations \ref{eqn:ebr} and \ref{eqn:bh}, one arrives at

\begin{equation}
    \nabla\times \vec{E}(\vec{r},t) + \mu_0 
    \frac{\partial \vec{H}(\vec{r},t)}{\partial t}=0
    \label{eqn:eh_timedomain}
\end{equation}

Using Equations \ref{eqn:hdj} and \ref{eqn:de}, we arrive at
\begin{equation}
    \nabla\times \vec{H}(\vec{r},t)- \epsilon(\vec{r}) \frac{\partial \vec{E}(\vec{r},t)}{\partial t}= \vec{J}
    \label{eqn:hej_timedomain}
\end{equation}

Let us fit the wave function in these equations. Hence, we have the following relationships.  $\vec{H}(\vec{r},t)= \vec{H}(\vec{r})e^{-i\omega t}$ and $\vec{E}(\vec{r},t)= \vec{E}(\vec{r})e^{-i\omega t}$, where $\omega$ is the frequency of the light wave. 

Therefore, in the frequency domain, Equations \ref{eqn:eh_timedomain} and \ref{eqn:hej_timedomain} can be interpreted as follows.

\begin{equation}
    \nabla\times \vec{E}(\vec{r},t) - i\omega \mu_0 \vec{H}(\vec{r},t) =0
    \label{eqn:eh}
\end{equation}
\begin{equation}
     \nabla\times \vec{H}(\vec{r},t)+  i\omega\epsilon(\vec{r}) \vec{E}(\vec{r},t)= \vec{J}
     \label{eqn:hej}
\end{equation}

Equations \ref{eqn:eh} and \ref{eqn:hej} lead to the Principal Equation for EM waves.
Let us use Equation \ref{eqn:eh} to find $\vec{H}(\vec{r},t)$ in terms of $\vec{E}(\vec{r},t)$.

\begin{equation}
            \frac{1}{\mu_0 (i\omega)}\nabla\times \vec{E}(\vec{r},t) = \vec{H}(\vec{r},t) 
\end{equation}

Let us now use Equation~\ref{eqn:hej}.

\begin{equation}
    \frac{1}{\mu_0 i \omega} \nabla \times \nabla \times \vec{E}(\vec{r},t) + \epsilon(\vec{r}) (i\omega) \vec{E}(\vec{r},t) = \vec{J} 
    \label{eqn:EfromJ_derive}
\end{equation}

The principal equation can thus be derived.
\begin{equation}
    \left ( \left ( \nabla \times \frac{1}{\mu_0}\nabla \times \right ) - \omega^2\epsilon(\vec{r}) \right ) \vec{E}(\vec{r},t)  = i\omega \vec{J} 
    \label{eqn:efromj}
\end{equation}

In the sections that follow, we shall denote $\vec{E}(\vec{r},t)$ and $\vec{H}(\vec{r},t)$ by $\vec{E}$ and $\vec{H}$ respectively. 

\section{Realizing a Dot Product in the Photonic Domain}

\subsection{Wave Equation for Inhomogeneous Media} 

\begin{figure}[!htb]
\centering
\includegraphics[width=0.5\linewidth]{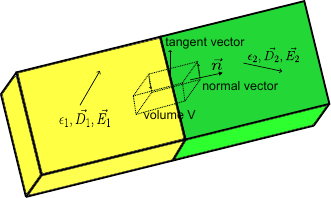}
\caption{Wave propagation at the medium boundary}
\label{fig:boundary_cond}
\end{figure}

Let us consider the case where there are no current sources and
there is no free charge present in the waveguide. 
This means that $J=0$ and $\rho=0$. 
Therefore, Equation \ref{eqn:drho} becomes $\nabla.\vec{D}=0$. This implies that
$\nabla.(\epsilon(\vec
{r})\vec{E}(\vec{r}))=0$ (from Equation \ref{eqn:de}).

Next, let us consider Figure~\ref
{fig:boundary_cond} that shows two different media juxtaposed with each other. Their permittivities are $\epsilon_1$ and $\epsilon_2$, respectively. 

Let us use the divergence theorem, $ \int_{V} \nabla.\vec{D} dV = \oint_S (\vec{D}.\vec{n}) dS, $ where $dS$ is the surface element of surface $S$, whose unit normal vector is $\vec{n}$, which bounds volume $V$. This implies $(\vec{D_2}.\vec{n} + \vec{D_1}.(-\vec{n}))S=0.$ Thus, the normal component of the displacement field should be continuous at the boundary between two dielectric media \cite{MaterialBoundaries}. 

Given that there is no current flow, Equation~\ref{eqn:efromj} becomes

\begin{equation}
    ((\nabla \times \frac{1}{\mu_0}\nabla \times)- \omega^2\epsilon(\vec{r}))\vec{E}=0
    \label{eqn:ewithoutj}
\end{equation}
$\because \vec{E}= \vec{E}(\vec{r}) e^{-i \omega t}, 
\therefore \frac{\partial^2 \vec{E}}{\partial t^2}= - \omega^2 \vec{E}(\vec{r})e^{-i \omega t} = -\omega^2 \vec{E}.$ Thus, the term $-\omega^2 \vec{E}$ can be replaced by $ \frac{\partial^2 \vec{E}}{\partial t^2}$ in Equation \ref{eqn:ewithoutj} if it needs to be shown in the time domain. 

Hence, in the time domain, Equation~\ref{eqn:ewithoutj} becomes the following.

\begin{equation}
    \left ( (\nabla \times \nabla \times)+ \mu_0\epsilon(\vec{r})\frac{\partial^2}{\partial t^2} \right ) \vec{E}=0
    \label{eqn:ewithoutj_timedomain}
\end{equation}
Let us denote the wavefunction of an optical wave by $u(\vec{r},t),$ where $\vec{r}$ is the position vector and $t$ is the time. Let $c_0= \frac{1}{\sqrt{\epsilon_0\mu_0}}=3 \times 10^8 m/s $ be the speed of light in free space. Optical waves satisfy the wave equation \cite{doi:https://doi.org/10.1002/0471213748.ch2}.

\begin{equation}
    \nabla^2u-\frac{1}{c_0^2}\frac{\partial^2u}{\partial t^2}=0
    \label{eqn:wave_eqn}
\end{equation}

Since the wave equation (Eqn. \ref{eqn:wave_eqn}) is linear, that is, if $u_1$ and $u_2$ are solutions of the wave equation, then $u= u_1+u_2$ is also a solution of the equation, the superposition principle holds for optical waves. 

Let us now fit Equation~\ref{eqn:ewithoutj_timedomain} in the form of a wave equation~\cite{doi:https://doi.org/10.1002/0471213748.ch5}. 

We shall use the following result from vector calculus, where $\vec{v}$ is a vector.
\begin{equation}
    \nabla \times (\nabla \times \vec{v})= \nabla(\nabla.\vec{v})-\nabla^2 \vec{v} 
    \label{eqn:curl_of_curl_vec_identity}
\end{equation}

Equation~\ref{eqn:ewithoutj_timedomain} can be rewritten as 
\begin{equation}
    \nabla(\nabla.\vec{E})-\nabla^2 \vec{E} + \mu_0\epsilon(\vec{r})\frac{\partial^2 \vec{E}}{\partial t^2}=0
    \label{eqn:vec_identity_in_ewithoutj}
\end{equation}

Let us now derive the value of $\nabla.\vec{E}(\vec{r},t)$. 

\begin{equation}
    \begin{split}
 \nabla.(\epsilon(\vec{r}) \vec{E}(\vec{r},t))=0 
 \quad (Eqn.\ref{eqn:drho})  \\
\implies \vec{E}(\vec{r},t)\nabla.\epsilon(\vec{r},t) + \epsilon(\vec{r})\nabla.\vec{E}(\vec{r},t)=0
    \\ \implies \nabla.\vec{E}(\vec{r},t) =  \frac{-\vec{E}(\vec{r},t)\nabla.\epsilon(\vec{r})}{\epsilon(\vec{r})}
\end{split}
\label{eqn:nablae}
\end{equation}

The refractive index $n(\vec{r})$ is related to electrical permittivity by the relation $n(\vec{r})= \sqrt{\frac{\epsilon(\vec{r})}{\epsilon_0}}$. The speed of light $c$ in a medium of refractive index $n$ is $\frac{c_0}{n}= \sqrt{\frac{1}{\mu_0\epsilon}}$, where $n= \sqrt{\frac{\epsilon}{\epsilon_0}}$. In a medium composed of multiple refractive indices, the speed of light becomes a function of the position vector $\vec{r}$. It can be derived as follows. Let us denote the speed of light in varying dielectric media by $c(\vec{r})$.
\begin{equation}
    \begin{split}
        \epsilon(\vec{r}) = \epsilon_0 n(\vec{r})^2 \\
      \implies  c(\vec{r})= \frac{c_0}{n(\vec{r})}\\
        \implies n(\vec{r})= \frac{\sqrt{\frac{1}{\mu_0 \epsilon_0}}} {c(\vec{r})} 
    \end{split}
\end{equation}

\begin{equation}
    c(\vec{r})^2= \frac{1}{\mu_0 \epsilon(\vec{r})}
    \label{eqn:c_r_derive}
\end{equation}

Using Equations~\ref{eqn:vec_identity_in_ewithoutj} and \ref{eqn:nablae} we get

\begin{equation}
    \nabla^2 \vec{E} - \frac{1}{c(\vec{r})^2}\frac{\partial^2 \vec{E}}{\partial t^2}= \nabla \left ( \frac{-\vec{E} \nabla.\epsilon(\vec{r})}{\epsilon(\vec{r})} \right )
    \label{eqn:wave_eqn_inhomogeneous_media}
\end{equation}

The term in the $RHS$ of Equation~\ref{eqn:wave_eqn_inhomogeneous_media} cannot be neglected (approximated to $0$) because $\epsilon(\vec{r})$ undergoes abrupt changes with respect to the wavelength of light. However, Equation \ref{eqn:wave_eqn_inhomogeneous_media} is still linear in $\vec{E}$; therefore, the superposition principle holds. The superposition principle will be used to establish the cavity optimization framework in Section \ref{sec:obj_fn}.

\begin{stbox}
The section establishes the Principal Equation as a wave equation and shows that the superposition principle holds. The superposition principle will be used to support the optimization framework in later sections, which treat EM waves originating from multiple sources independently.
\end{stbox}

\subsection{Waveguide Wave Function}

Let us consider the dielectric distribution within a waveguide.
Assume $\epsilon(\vec{r})$ is a constant, where $\vec{r}$ is the
position vector. Equation~\ref{eqn:wave_eqn_inhomogeneous_media} becomes the following. 

\begin{equation}
     \nabla^2 \vec{E} - \frac{1}{c(\vec{r})^2}\frac{\partial^2 \vec{E}}{\partial t^2}=0
    \label{eqn:wave_eqn_homogeneous_media}
\end{equation}

Let $k$ be the wavenumber of the wave vector $||\vec{k}||,$ where $\vec{k}$ is the wave vector. $k= \frac{2\pi}{\lambda}= \frac{\omega}{c},$ where $\lambda$ is the wavelength, $c$ is the speed of light, and $\omega$ is the angular frequency. 

Recall the following relationships: 
$\vec{E}(\vec{r},t)= \vec{E}(\vec{r})e^{-i\omega t}, \frac{\partial^2 \vec{E}}{\partial t^2}= -\omega^2 \vec{E}(\vec{r}) e^{-i\omega t}= -\omega^2 \vec{E}(\vec{r},t).$ 
Equation \ref{eqn:wave_eqn_homogeneous_media} can be rewritten as 

\begin{equation}
\begin{split}
     \nabla^2 \vec{E}- \frac{1}{c^2}(-\omega^2 \vec{E})=0 \\ 
     \text{Putting } k= \frac{w}{c},\nabla^2 \vec{E}+ k^2 \vec{E}=0
\end{split}
   \label{eqn:helmhotz_equation}
\end{equation}

The above equation is called the Helmholtz equation. Both magnetic and electric fields satisfy it. Using the method of separation of variables and the Helmholtz equation, we find the solution to fields that exist in a waveguide. 

Assume that the wave propagates in the $x$ direction. Let us denote the electric fields $\vec{E}= \vec{E_t} + E_x\hat{x}$ and magnetic fields $\vec{H}= \vec{H_t} + H_x\hat{x}$, where $t$ denotes the axes normal to the $x$ axis while $\hat{x}$ denotes the $x$ axis \cite{Chew2016}.
$\vec{E}= (\vec{e}(y,z) + \hat{x}e_x(y,z))e^{-ik_x x},$ where $ \vec{e} (y,z)e^{-ik_x x}= \vec{E_t}, e_x(y,z)e^{-ik_x x}= E_x$. $k^2= k_x^2 + k_t^2,$ where $k_t^2=k_y^2+k_z^2$. Let us solve for $E_x$ using Equation~\ref{eqn:helmhotz_equation}. 

\begin{equation}
\begin{split}
    \nabla^2 E_x + k^2 E_x &= 0 \\
    \left( \frac{\partial^2}{\partial x^2} + \frac{\partial^2}{\partial y^2} + \frac{\partial^2}{\partial z^2} \right) E_x + k^2 E_x &= 0 \\
    \because \frac{\partial^2 E_x}{\partial x^2} = -k_x^2 e_x(y,z) e^{-ik_x x} &= -k_x^2 E_x \\
    \left( \frac{\partial^2}{\partial y^2} + \frac{\partial^2}{\partial z^2} \right) E_x - k_x^2 E_x + k^2 E_x &= 0 
\end{split}
\label{eqn:e_x_helmholtz_derive}
\end{equation}
From Equation~\ref{eqn:e_x_helmholtz_derive}, we arrive at the equation that solves $E_x$.
\begin{equation}
     \left( \frac{\partial^2}{\partial y^2} + \frac{\partial^2}{\partial z^2} \right) E_x + k_t^2 E_x  = 0 
     \label{eqn:tm_mode}
\end{equation}
In a similar fashion, we solve for $H_x$.
\begin{equation}
    \begin{split}
        \left ( \frac{\partial^2}{\partial y^2}+ \frac{\partial^2}{\partial z^2} \right )H_x + k_t^2 H_x=0
    \end{split}
     \label{eqn:te_mode}
\end{equation}
Equation~\ref{eqn:eh} says $\nabla \times \vec{E}(\vec{r})= i\omega\mu_0 \vec{H}(\vec{r})$ with the implication that the magnetic and electric fields are orthogonal to each other. Thus, if $H_x$ exists, $E_x$ is supposed to be 0 (transverse electric mode), and if $E_x$ were to exist, $H_x$ becomes 0 (transverse magnetic mode).  For transverse electric waves, we find $H_x$ through Equation~\ref{eqn:te_mode}.   

For transverse magnetic waves, we find $E_x$ through Equation~\ref{eqn:tm_mode}.  

From $E_x$ and $H_x$, the rest of the components of $\vec{E}$ and $\vec{H}$ can be found (the fields $E_t, H_t$).   

Using Equations \ref{eqn:eh} and \ref{eqn:hej} where $\vec{J}=0$, we get

\begin{equation}
    \begin{split}
    \left (\nabla_t+ \frac{\partial }{\partial x}\hat{x} \right ) \times (\vec{E_t} + E_x \hat{x}) &= i\omega\mu (\vec{H_t} + H_x \hat{x}) \\
\left (\nabla_t+ \frac{\partial }{\partial x}\hat{x} \right ) \times (\vec{H_t} + H_x \hat{x}) &= -i\omega\epsilon (\vec{E_t} + E_x \hat{x})
    \end{split}
    \label{eqn:separation_of_var_in_eqn_sans_J}
\end{equation}

Equating the transverse $t$ component in Equation \ref{eqn:separation_of_var_in_eqn_sans_J}, we get

\begin{equation}
    \nabla_t \times E_x \hat{x} + \frac{\partial}{\partial x}\hat{x} \times \vec{E_t}  =i \omega \mu \vec{H_t} 
    \label{eqn:h_t_term}
\end{equation}

\begin{equation}
         \nabla_t \times H_x \hat{x} + \frac{\partial}{\partial x}\hat{x} \times \vec{H_t} = -i \omega \epsilon \vec{E_t}
   \label{eqn:e_t_term}
\end{equation}

Let us substitute $E_t$ of Equation~\ref{eqn:e_t_term} into Equation~\ref{eqn:h_t_term}.

\begin{equation}
    \begin{split}
        \vec{H_t} &= \frac{1}{i\omega \mu} \left (\nabla_t \times E_x \hat{x}+ \frac{\partial}{\partial x} \hat{x} \times \left (\frac{1}{-i \omega \epsilon} \left (\nabla_t \times H_x \hat{x} + \frac{\partial}{\partial x} \hat{x} \times \vec{H_t} \right ) \right ) \right )
    \end{split}
    \label{eqn:ht_ex_hx}
\end{equation}

Let us use the following vector identities.

\begin{equation}
    \begin{split}
        \hat{x}\times \nabla_t \times \hat{x}= \nabla_t \\
        \hat{x}\times\hat{x}\times \vec{H_t}= - \vec{H_t}
    \end{split}
    \label{eqn:vector_rules}
\end{equation}

Let us solve Equation~\ref{eqn:ht_ex_hx}.

\begin{equation}
     \vec{H_t}= \frac{1}{i\omega\mu} \left (\nabla_t \times E_x \hat{x} + \frac{1}{-i \omega \epsilon} \left (\frac{\partial}{\partial x }\nabla_t H_x - \frac{\partial^2}{\partial x^2} \vec{H_t} \right ) \right ) 
\end{equation}

Since the $\vec{H_t}$ fields have a dependence on $e^{ik_x x}$, $\frac{\partial^2}{\partial x^2}$ leads to $-k_x^2$.

\begin{equation}
    \implies \vec{H_t} = \frac{1}{i\omega\mu} \left (\nabla_t \times E_x \hat{x} + \frac{1}{-i \omega \epsilon} \left ( \frac{\partial}{\partial x }\nabla_t H_x + k_x^2 \vec{H_t} \right ) \right)
    \label{eqn:h_t_intermediate}
\end{equation}

Let us consider $k$ in terms of $\omega$ and $c$, which yields $k= \frac{\omega}{c} = \omega \sqrt{\mu \epsilon}$.  

The equation governing transverse magnetic fields $H_t$ can be derived from Equation \ref{eqn:h_t_intermediate}.

\begin{equation}
\begin{split}
    \vec{H_t} &= \frac{1}{i\omega\mu} \left (- \hat{x} \times \nabla_t E_x + \frac{1}{-i \omega \epsilon} \left ( \frac{\partial}{\partial x }\nabla_t H_x + k_x^2 \vec{H_t} \right ) \right) \\
        &= \left (- \frac{1}{i\omega\mu} \hat{x} \times \nabla_t E_x + \frac{1}{\omega^2 \mu \epsilon} \left ( \frac{\partial}{\partial x }\nabla_t H_x + k_x^2 \vec{H_t} \right ) \right) \\
        &= \left (- \frac{1}{i\omega\mu} \hat{x} \times \nabla_t E_x + \frac{1}{k^2} \left ( \frac{\partial}{\partial x }\nabla_t H_x + k_x^2 \vec{H_t} \right ) \right) \\
    \implies \frac{k^2-k_x^2}{k^2} \vec{H_t} &= \left (- \frac{1}{i\omega\mu} \hat{x} \times \nabla_t E_x + \frac{1}{k^2} \left ( \frac{\partial}{\partial x }\nabla_t H_x  \right ) \right) \\
    \implies \vec{H_t} &= \frac{1}{k^2-k_x^2} \left ( \frac{ik^2}{\omega\mu} \hat{x} \times \nabla_t E_x +  \frac{\partial}{\partial x }\nabla_t H_x  \right )
\end{split}
\label{eqn:h_t_term_final}
\end{equation}

\begin{equation}
    \vec{H_t}= \frac{1}{(k^2-k_x^2)} \left (\frac{\partial}{\partial x}\nabla_t H_x + i \omega \epsilon \hat{x}\times \nabla_tE_x \right )
    \label{eqn:transverse_H}
\end{equation}

Following the above approach, we arrive at $\vec{E_t}$ when substituting $\vec{H_t}$ of Equation \ref{eqn:h_t_term} into Equation \ref{eqn:e_t_term}.

\begin{equation}
     \vec{E_t}= \frac{1}{(k^2-k_x^2)} \left ( \frac{\partial}{\partial x}\nabla_t E_x - i \omega \mu \hat{x}\times \nabla_t H_x \right ) 
       \label{eqn:transverse_E}
\end{equation}

Let us derive the boundary conditions \cite{Chew2016}. 
\begin{figure}[!htb]
\centering
\includegraphics[width=0.5\linewidth]{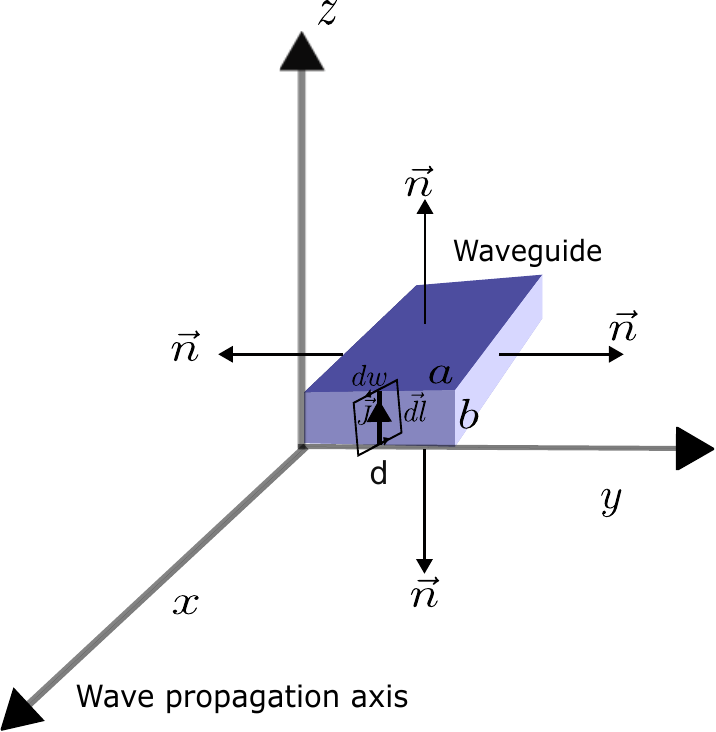}
\caption{Wave propagation at the boundary of the waveguide}
\label{fig:current_source}
\end{figure}

Let us go back to Equation~\ref{eqn:ebr} to establish boundary conditions for electric fields. Let us refer to Figure \ref{fig:material_surface}.  Equation~\ref{eqn:ebr} shows $\int_C \vec{E}.\vec{dl}= -\iint_S(\frac{\partial \vec{B}}{\partial t}.\vec{dS}),$ where $\vec{dS}$ is the area vector of the infinitesimal piece of surface bounded by C. Let $dw$ be the width of the closed loop $C$, tending to $0$. Thus, the area encompassed by $dS$ is 0.  Therefore, $\int_C \vec{E}.\vec{dl}=0$, with the implication that the tangential component of the electric field must be constant across the boundary of any two materials. Suppose that the waveguide is metallic. There cannot be any tangential component of $\vec{E}$ present on a metallic surface. Thus, the boundary condition becomes that there is no tangential component of $\vec{E}$ present at the surface. 

\begin{figure}[!htb]
\centering
\includegraphics[width=0.5\linewidth]{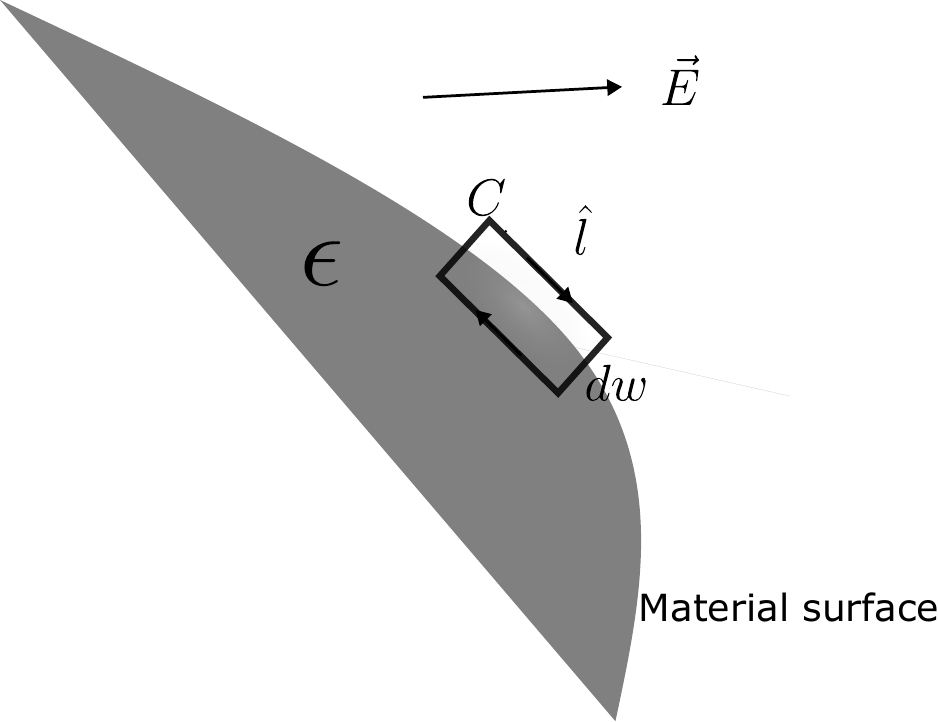}
\caption{Wave propagation at the material surface }
\label{fig:material_surface}
\end{figure}

For the transverse electric mode (TEM), $E_x=0$. To ensure that $\vec{E_t}$ is $0$ on the surface, let us derive the restrictions on the fields on the surface.
Plugging $E_x=0$ into Equation~\ref{eqn:transverse_E} results in the following.

\begin{equation}
    \vec{E_t}= \frac{1}{(k^2-k_x^2)}(-i\omega\mu\hat{x}\times \nabla_t H_x) 
    \label{eqn:e_t}
\end{equation}

$\vec{E_t}$ becomes $0$ when $\hat{x} \times \nabla_t H_x = 0$. Let us refer to Figure \ref{fig:current_source} where $\vec{J}$ is considered $0$. The condition is equivalent to $\hat{n}.\nabla_t H_x = \frac{\partial}{\partial n} H_x=0$. 
For the transverse magnetic mode (TM), $H_x=0$. In this case, $ \frac{\partial}{\partial n} H_x=0$ on the surface. $\vec{E_t}$ at the surface becomes $\vec{E_t}= \frac{1}{(k^2-k_x^2)} \left (\frac{\partial}{\partial x}\nabla_t E_x \right ) $. The boundary condition for the TM mode becomes $E_x=0$ for $\vec{E_t}$ to become $0$ on the surface of the material. 

\begin{stbox}
The section solves the wave equation for a rectangular waveguide along with the boundary conditions of the magnetic and electric fields with the assumption being that the waveguide is metallic. It shows that the derivative of $H_x$ with respect to the normal of the surface is $0$ for a transverse electric mode (Figure \ref{fig:current_source}). The boundary condition tells us that $H_x$ should be of the form $\cos(\frac{m\pi}{a}y)\cos(\frac{n\pi}{b}z)$. The structure ensures that differentiating $H_x$ with respect to $y$ and $z$ produces a sine term that becomes $0$ whenever $y$ and $z$ reach the extremes of the waveguide, respectively. On the other hand, a transverse magnetic mode requires $E_x$ to become $0$ on the surface. This brings out the structure of $E_x$. It needs to be of the form $\sin(\frac{m\pi}{a}y)\sin(\frac{n\pi}{b}z)$. At the extremes of the waveguide along $y$ and $z$, sine terms ensure $E_x$ becomes $0$.
\end{stbox}

\subsection{Activating Waveguide Modes Using $\vec{J}$}

Let us list the equations that govern the TE and TM modes.

Transverse Electric Mode:
\begin{equation}
\begin{split}
     (\nabla_t^2 + k_t^2)H_x=0  \\ \frac{\partial}{\partial n}H_x=0 \text{ on surface}
\end{split}  
\label{eqn:TE_mode_condition}
\end{equation}
Transverse Magnetic Mode:
\begin{equation}
\begin{split}
     (\nabla_t^2 + k_t^2)E_x=0  \\ E_x=0 \text{ on surface}
\end{split} 
\label{eqn:TM_mode_condition}
\end{equation}

Let the waveguide be rectangular. 

\begin{figure}[!htb]
\centering
\includegraphics[width=0.5\linewidth]{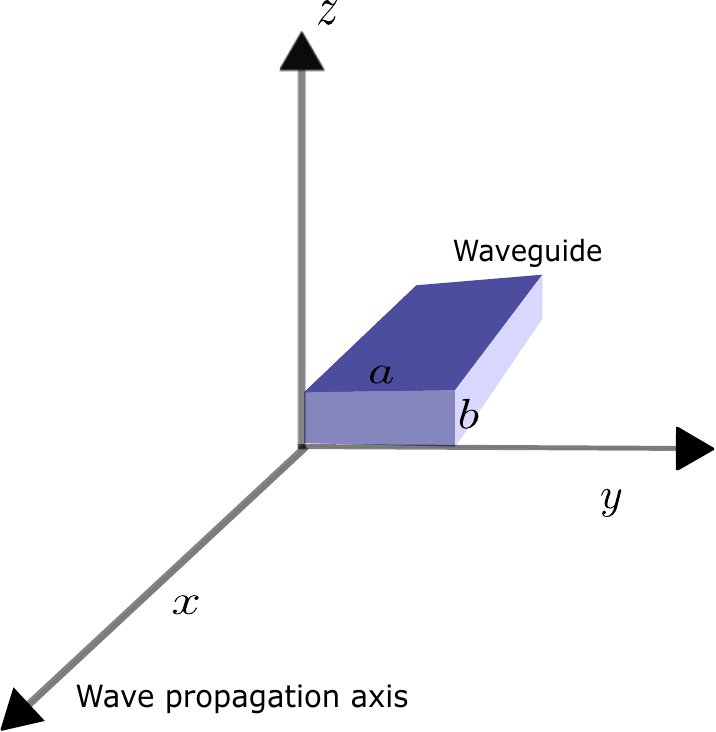}
\caption{Waveguide}
\label{fig:rect_waveguide}
\end{figure}

Let us look at the field patterns as the wave propagates in the waveguide (in Figure \ref{fig:rect_waveguide}) in the $x$ direction.
The field patterns satisfying the conditions of a transverse electric mode (see Eqn.~\ref{eqn:TE_mode_condition}) are as follows: 

\begin{equation}
\begin{split}
     H_x= H_{mn} \cos \left ({\frac{m\pi}{a}y} \right )\cos \left ( {\frac{n\pi}{b}z} \right ) e^{ik_xx}   \text{ with }    k_t^2= \frac{m\pi}{a}^2 + \frac{n\pi}{b} ^2
\end{split}
    \label{eqn:TEM_in_rect_waveguide}
\end{equation}

Similarly, the field patterns satisfying the transverse magnetic mode (Eqn.~\ref{eqn:TM_mode_condition}) are the following.

\begin{equation}
\begin{split}
     E_x= E_{mn} \sin \left ( {\frac{m\pi}{a}y} \right )\sin \left ({\frac{n\pi}{b}z} \right )e^{ik_xx}   \text{ with }    k_t^2= \frac{m\pi}{a}^2 + \frac{n\pi}{b} ^2
\end{split}
    \label{eqn:TM_in_rect_waveguide}
\end{equation}

$H_{mn}$ and $E_{mn}$ are the magnetic field and electric field amplitudes, respectively. $a$ is the dimension of the waveguide in the $y$ direction; $b$ is the dimension of the waveguide in the $z$ direction. $m$ and $n$ are whole numbers that identify the field patterns that can exist in the system (also known as mode numbers). The mode numbers are connected to the medium's electric permittivity by the following relation.

Let us revisit the equation for $k$.

\begin{equation*}
\begin{split}
      k^2 &= k_x^2 + k_t^2 \\
      k^2 &= \omega^2 \mu \epsilon \\
      \implies k_x^2 &= \omega^2 \mu \epsilon  - \left (\frac{m\pi}{a}^2 + \frac{n\pi}{b} ^2 \right )
\end{split}
\end{equation*}

If $k_x^2$ falls below $0$, $k_x$ has an imaginary component. The $e^{ik_x x}$ term in $H_x$ becomes $e^{-|k_x|x},$ therefore, the mode becomes non-propagating. Thus, for the mode to exist, $k_x^2 \ge 0$. 
The cutoff frequency for the mode to exist becomes the following. 

\begin{equation*}
\begin{split}
\omega_{m,n,\epsilon} = \frac{1}{\sqrt{\mu \epsilon}} \sqrt{\frac{m\pi}{a}^2 + \frac{n\pi}{b} ^2}
\end{split}
\label{eqn:cutoff_freq_for_propagating_mode}
\end{equation*}

Let us inject a current source $\vec{J}$ now.
Let us use Equation \ref{eqn:efromj}. 
We shall be using the following results.

\begin{equation}
     ((\nabla \times \nabla \times)- k^2)\vec{E}=(i\omega \mu \vec{J}) 
    \label{eigenmode_eqn:efromj}
\end{equation}

If two functions, $f$ and $g$, are orthonormal, then the following is true.

\begin{equation}
    \int g^*(r)f(r)dr = \begin{cases}
        1 & \text{if $g(r)= f(r)$} \\
        0 & \text{otherwise }
    \end{cases}
    \label{eqn:orthonormal_vec_property}
\end{equation}

Let us define the Dirac delta function $\delta (x)$.

\begin{equation}
    \begin{split}
    & \int_{-\infty}^\infty \delta(x) = 1 \\
    & \delta(x) = 0 \quad (x \ne 0) \\
    \end{split}
    \label{eqn:dirac_delta}
\end{equation}

Let us represent $\vec{E}$ in terms of its eigenvectors and eigenvalues. They are calculated after solving the electric field eigenvalue equation (Eqn.~\ref{eqn:ewithoutj}). Let $\vec{F_m}$ be the eigenvector and $k_m$ be the eigenvalue of $\vec{E}$. 
Let us use the cross-product rule of vector calculus, where $\vec{A}$ and $\vec{B}$ are vector fields. 

\begin{equation}
    \begin{split}
        (\vec{A} \times \nabla).\vec{B} = \vec{A} . (\nabla \times \vec{B}) 
    \end{split}
\end{equation}

Let us prove that the operator $\nabla \times \nabla \times$ is a Hermitian operator. An operator $A$ is Hermitian if the following holds for any two vector fields $\vec{u}$ and $\vec{v}$: $(A\vec{u},\vec{v})=(\vec{u},A\vec{v}),$ where $(F,G)$ denotes the inner product of any two vector fields $\vec{F}$ and $\vec{G}$= $\int \vec{F}^*(\vec{r})\vec{G}(\vec{r})dr$. 
\begin{align*}
    (\vec{u}, \nabla \times \nabla \times \vec{v}) = \int \vec{u}^*.(\nabla \times \nabla \times \vec{v})dr \\
    =  \int (\nabla \times \vec{u})^*.(\nabla \times \vec{v}) dr \\
    = \int (\nabla \times \nabla \times \vec{u})^*. \vec{v} dr \\
    = (\nabla \times \nabla \times \vec{u}, \vec{v})
\end{align*}

Thus, the operator $\nabla \times \nabla \times$ is indeed Hermitian, which means Equation \ref{eigenmode_eqn:efromj}, where $\vec{J}=0$, produces orthogonal eigenvectors with real-valued eigenvalues. The electric field can be represented as $E(r)= \sum_m a_m \vec{F_m}(\vec{r})$, where $\vec{F_m}$ satisfies Equation \ref{eqn:efromj}: 

\begin{equation}
     ((\nabla \times \nabla \times)- k_m^2)\vec{F_m}(\vec{r})= 0
    \label{eigenmode_eqn:Fm_field_vec}
\end{equation}

The $a_m$ values can be estimated using the following equations. These values indicate which modes $\vec{F_m}$ are activated in the system after injecting the stimulus $\vec{J}$. 

Using Equations \ref{eigenmode_eqn:efromj} and \ref{eigenmode_eqn:Fm_field_vec}, we arrive at the following.

\begin{equation*}
    \begin{split}
        \nabla \times \nabla \times \sum_m a_m \vec{F_m}(\vec{r})- k^2 \sum_m a_m \vec{F_m}(\vec{r}) = i\omega\mu \vec{J}(\vec{r}) \\
       \implies  \sum_m a_m (k_m^2 - k^2) \vec{F_m} (r) = i\omega \mu \vec{J}(\vec{r}) 
    \end{split}
\end{equation*}

Using the orthonormal property \ref{eqn:orthonormal_vec_property} of $\vec{F_m}$, we arrive at the estimate of $a_m$.
\begin{equation*}
    \begin{split}
       a_m= i \omega \mu \frac{\int \vec{F_m}^*\vec{J}(\vec{r})dr}{k_m^2-k^2}
    \end{split}
\end{equation*}

\begin{stbox}
    The section shows that the operator $\nabla \times \nabla \times$ is Hermitian, which implies that the electric and magnetic fields can be represented in terms of orthogonal eigenvectors and real-valued eigenvalues. The orthogonal property is used to arrive at an expression that determines the distribution of modes of propagation upon being fed by a stimulus in the form of a current. This will eventually serve as the basis for encoding data. 
\end{stbox}

\subsection{Data Encoding Using $\vec{J}$}

Let us define the source of current $\vec{J_s}$ as $\vec{J_s}= J_0 \delta(y-d)\hat{z}$, where $\delta(y)$ is the Dirac delta function, $d$ is the position along the $y$ axis, where the current source is placed, $\hat{z}$ is the axis of the current flow, and $J_0$ measures the power emitted by the source. Here, $\vec{J_s}$ is a current sheet, thus confined to the surface. We can refer to Figure \ref{fig:current_source}, where $\vec{J}$ is the current sheet $\vec{J_s}$. We place $\vec{J}$ so that it activates the fundamental mode of the system (the mode with the lowest cutoff frequency). Electric fields are induced by static charges. In our case, we have the current as the source, and thus there are no static charges. Therefore, magnetic fields are generated in directions perpendicular to the direction of current flow (Equation \ref{eqn:hdj}). Magnetic fields in turn induce electric fields that are perpendicular to the former, which implies that there will be no $x$ component in $\vec{E}$. Thus, TE modes get excited. 

The boundary conditions for magnetic fields when a current source is placed can be derived from Equations \ref{eqn:eh}, \ref{eqn:hej} and \ref{eqn:efromj} as follows \cite{Chew2016}.

Let us refer to Equation \ref{eqn:efromj} and the closed loop in Figure~\ref{fig:current_source}.

\begin{equation*}
    \iint_S (\nabla \times \frac{1}{\mu_0} \nabla  \times \vec{E}).d\vec{S} - \omega^2 \iint_S (\epsilon(\vec{r})\vec{E}).d\vec{S} = i \omega \iint_S \vec{J}.d\vec{S} 
\end{equation*}

Invoking the Stokes' theorem, we arrive at the following.

\begin{equation*}
    \oint_C (\frac{1}{\mu_0} \nabla \times \vec{E}).d\vec{l} - \omega^2  \iint_S (\epsilon(\vec{r})\vec{E}).d\vec{S} = i \omega \iint_S \vec{J}.d\vec{S}
\end{equation*}

Let us assume that we have a sheet of current such that $dw$ tends to $0$, therefore, $ d\vec{S}$ also becomes 0.

\begin{equation*}
    \begin{split}
        & \implies \oint_C \left (\frac{1}{\mu_0} \nabla \times \vec{E} \right ).d\vec{l}  = i \omega \vec{J_s} \\
& \implies \vec{x} \times \left (\frac{1}{\mu_0} \nabla \times \lim_{x \rightarrow 0^+}\vec{E(x)} - \frac{1}{\mu_0} \nabla \times \lim_{x \rightarrow 0^-}\vec{E(x)} \right ) = i \omega \vec{J_s}
    \end{split}
\end{equation*}

Using Equation \ref{eqn:eh}, we get the following.

\begin{equation*}
    \vec{x} \times i \omega \left (\lim_{x \rightarrow 0^+}\vec{H(x)} -  \lim_{x \rightarrow 0^-}\vec{H(x)} \right ) = i \omega \vec{J_s}
\end{equation*}

\begin{equation}
    \vec{x} \times \left (\lim_{x \rightarrow 0^+}\vec{H(x)} -  \lim_{x \rightarrow 0^-}\vec{H(x)} \right ) =  \vec{J_s}
     \label{data_encoding_eqn:boundary_cond}
\end{equation}

Since $\vec{J}$ has no component along the $\hat{y}$ direction, Equation \ref{data_encoding_eqn:boundary_cond} indicates that $\lim_{x \rightarrow 0^{+}}\vec{H}(x)- \lim_{x \rightarrow 0^{-}}\vec{H}(x)$ must have no $\hat{z}$ component. 
From Equations \ref{eqn:transverse_H} and \ref{eqn:TEM_in_rect_waveguide}, the following deductions can be made. 
Whenever $n \neq 0$, $H_{mn}$ should be $0$ to ensure $\vec{H}$ is $0$ in the $\hat{z}$ direction. This also implies that $H_x$ becomes $0$, resulting in $\vec{H}$ being $0$ in the $\hat{y}$ direction. Thus, we consider $n=0$ because the fields $\vec{H}$ cease to exist for any other value of $n$. 


Using Equation \ref{eqn:TEM_in_rect_waveguide}, we arrive at the following.

\begin{equation}
\begin{split}
    \lim_{x\rightarrow 0^+} H_x &=  \sum_m H_{m0} \cos \left ({\frac{m\pi}{a}y} \right ) e^{ik_x(0+x)} \\
    &= \sum_m H_{m0} \cos \left ({\frac{m\pi}{a}y} \right ) e^{ik_x(x)}\\
   \lim_{x\rightarrow 0^-} H_x &=  \sum_m H_{m0} \cos \left ({\frac{m\pi}{a}y} \right ) e^{ik_x(0-x)} \\
   & =  \sum_m H_{m0} \cos \left ({\frac{m\pi}{a}y} \right ) e^{-ik_x(x)}
\end{split}
\label{eqn:lim_Hx}
\end{equation}

Let us feed $H_x$ of Equation~\ref{eqn:lim_Hx} into Equation~\ref{eqn:transverse_H}. 

\begin{equation}
\begin{split}
  \lim_{x \rightarrow 0^+}\vec{H_t} &= \sum_m \frac{1}{(k^2-k_x^2)}ik_x \left (-H_{m0}sin \left (\frac{m\pi}{a}y \right )\frac{m\pi}{a} \right ) e^{ik_xx} \hat{y} \\
  \lim_{x \rightarrow 0^-}\vec{H_t} &= \sum_m \frac{1}{(k^2-k_x^2)}(-ik_x) \left (-H_{m0}sin \left (\frac{m\pi}{a}y \right )\frac{m\pi}{a} \right ) e^{-ik_xx} \hat{y}
\end{split}
    \label{eqn:Transverse_H_for_J_encoding}
\end{equation}

Let us use Figure \ref{fig:current_source} as a reference point. Since $\vec{E_t}$ is supposed to remain constant across the boundary, the term $H_{m0}$ must also remain constant across the boundary $\quad (Eqn. \ref{eqn:transverse_E})$. Using $H_t$ from Equation \ref{eqn:Transverse_H_for_J_encoding} in Equation \ref{data_encoding_eqn:boundary_cond}, we arrive at the following. 

\begin{equation}
    \begin{split}
        \frac{1}{(k^2-k_x^2)}ik_x \frac{m\pi}{a}\sum_m \hat{x} \times (lim_{x \rightarrow 0^+} \left (-H_{m0}sin \left (\frac{m\pi}{a}y \right )\frac{m\pi}{a} e^{ik_x(x)} \right ) \hat{y} - \\
       lim_{x \rightarrow 0^-} \left (H_{m0}sin \left (\frac{m\pi}{a}y \right )\frac{m\pi}{a} e^{-ik_x(x)} \right ) \hat{y} ) = \vec{J_s} \\
       \implies - 2\sum_m ik_xH_{m0} \frac{1}{k^2-k_x^2}\frac{m\pi}{a}sin \left (\frac{m\pi}{a} y \right ) \hat{z} = J_0 \delta(y-d) \hat{z}
    \end{split}
    \label{eqn:j_exciting_modes_derive}
\end{equation}

From the derivation in Equation \ref{eqn:j_exciting_modes_derive}, we arrive at the equation that connects $H_{m0}$ to the current $\vec{J_s}$.

\begin{equation}
    \begin{split}
       - 2\sum_m ik_xH_{m0} \frac{1}{k^2-k_x^2}\frac{m\pi}{a}sin(\frac{m\pi}{a} y)= J_0 \delta(y-d)
    \end{split}
    \label{eqn:J_encoding_Exciting_modes}
\end{equation}

The lowest value that $m$ can assume is 1 (if $m=0,$ the term $k^2-k_x^2=k_t^2= \left (\frac{m\pi}{a}^2 + \frac{n\pi}{b}^2 \right )$ becomes 0 (since $n$ is also zero), which will make the fields $\vec{H_t}$ undefined. 
Let us denote the term $\frac{-2ik_xH_{m0}}{k^2-k_x^2}\frac{m\pi}{a}$ by $c_m$. Equation \ref{eqn:J_encoding_Exciting_modes} becomes the following.

\begin{equation}
    \sum_{m=1}^\infty c_m sin(\frac{m\pi}{a}y) = J_0 \delta(y-d)
\end{equation}

Let us determine $c_m$ using Fourier series analysis.

Let us use the following orthogonality relation of sine waves \cite{2021Fourier}, where $m$ and $n$ are integers.

\begin{equation}
    \int_{-L}^{L} sin \left (\frac{m\pi x}{L} \right ) sin \left (\frac{n\pi x}{L} \right )dx = L \delta(n-m) 
    \label{eqn:fourier_sine_series_ortho_condition}
\end{equation}

Using equation \ref{eqn:fourier_sine_series_ortho_condition} to determine $c_m$:
Let us define a function $f(y)= \sum_{m=1}^\infty c_m sin \left (\frac{m\pi}{a}y \right )$, and
a positive integer $n$.

\begin{equation}
    \begin{split}
        \int_0^a f(y)sin \left (\frac{n\pi y}{a} \right )dy =  \sum_{m=1}^\infty c_m \int_0^a sin \left (\frac{m\pi}{a}y \right )sin \left (\frac{n\pi}{a}y \right ) dy \\
        \implies \int_0^a f(y)sin \left (\frac{n\pi y}{a} \right )dy =  \sum_{m=1}^\infty c_m \frac{1}{2}\int_{-a}^a sin \left (\frac{m\pi}{a}y \right )sin \left (\frac{n\pi}{a}y \right ) dy  
    \end{split}
\end{equation}

Let us use Equation \ref{eqn:fourier_sine_series_ortho_condition}.

\begin{equation}
     \int_0^a f(y)sin \left (\frac{n\pi y}{a} \right )dy =  \sum_{m=1}^\infty c_m \frac{1}{2} a \delta(n-m) = a\frac{c_n}{2} 
\end{equation}

Therefore, $c_n$ becomes the following.

\begin{equation}
    \begin{split}
        c_n &= \frac{2}{a} \int_0^a f(y)sin \left (\frac{n\pi y}{a} \right )dy \\
           \implies c_n &= \frac{2}{a} \int_0^a J_0 \delta(y-d) sin \left (\frac{n\pi y}{a} \right )dy 
    \end{split}
\end{equation}

\begin{equation}
     c_n= \frac{2}{a} J_0 sin \left (\frac{n\pi d}{a} \right )
     \label{eqn:Fourier_series_analysis}
\end{equation}

Let us use Equation \ref{eqn:Fourier_series_analysis} to find $H_{m0}$.

\begin{equation}
    H_{m0}= - \frac{k^2-k_x^2}{2ik_x}\frac{a}{m \pi}\frac{2}{a} J_0 sin \left (\frac{m\pi d}{a} \right ) = - \frac{k^2-k_x^2}{ik_x}\frac{J_0}{m \pi} sin \left (\frac{m\pi d}{a} \right )
    \label{eqn:TE_mode_excitation_coefficient}
\end{equation}

If the fundamental mode needs to be activated, then $m=1$. $H_{10}$ peaks when the current source is placed at $d= \frac{a}{2}$; thus indicating that the placement and magnitude of the current source directly affect the magnetic field distribution in the system. 

Equation \ref{eqn:TE_mode_excitation_coefficient} shows that the excitation coefficient $H_{mn}$ is directly proportional to the amplitude $J_0$ of the light source $\vec{J}$. Equation \ref{eqn:TEM_in_rect_waveguide} shows that the amplitude of the magnetic field is $|H_{mn}|$. Moreover, switching the direction of $\vec{J}$ also switches the sign of the coefficient $H_{mn}$ (Equation \ref{eqn:TE_mode_excitation_coefficient}), in turn, switching the direction of magnetic fields (Equation \ref{eqn:TEM_in_rect_waveguide}). Equations \ref{eqn:transverse_E}, \ref{eqn:transverse_H} show that the $H_x$ fields that are estimated through the calculation of $H_{mn}$ (from Equations \ref{eqn:TE_mode_excitation_coefficient} and \ref{eqn:TEM_in_rect_waveguide}) directly affect the electric field. Therefore, the data can be encoded in terms of the magnitude and direction of the source $\vec{J}$, as any changes in its direction and magnitude will be reflected in the electric and magnetic fields generated in the system. 

\begin{stbox}
    The section presents Equation \ref{eqn:TE_mode_excitation_coefficient} as the basis for encoding the data on which the rest of the sections are based. Since we will be dealing with the fundamental mode of propagation, we will keep $m$ as $1$. The equation shows a direct relationship between the amplitudes of the magnetic field and the current source. Furthermore, the equation shows that reversing the direction of current reverses the phase of fields as well. Thus, the current source can be used to carry information about the data.  
\end{stbox}


\subsection{Optical Circuit}
Let us say $x$ and $y$ are two numbers whose product has to be calculated in the optical domain. The principle behind the optical vector dot products \cite{zhu2023lighteningtransformerdynamicallyoperatedopticallyinterconnectedphotonic} states that if an input pair $(x,y)$ is encoded as a light wave with wavelength $\lambda$, then the following equations hold. 

Let $P_o$ be the output port. 
$P^r_o$ be the right output port and $P^l_o$ be the left output port. 

\begin{figure}[!htb]
\centering
\includegraphics[width=0.5\linewidth]{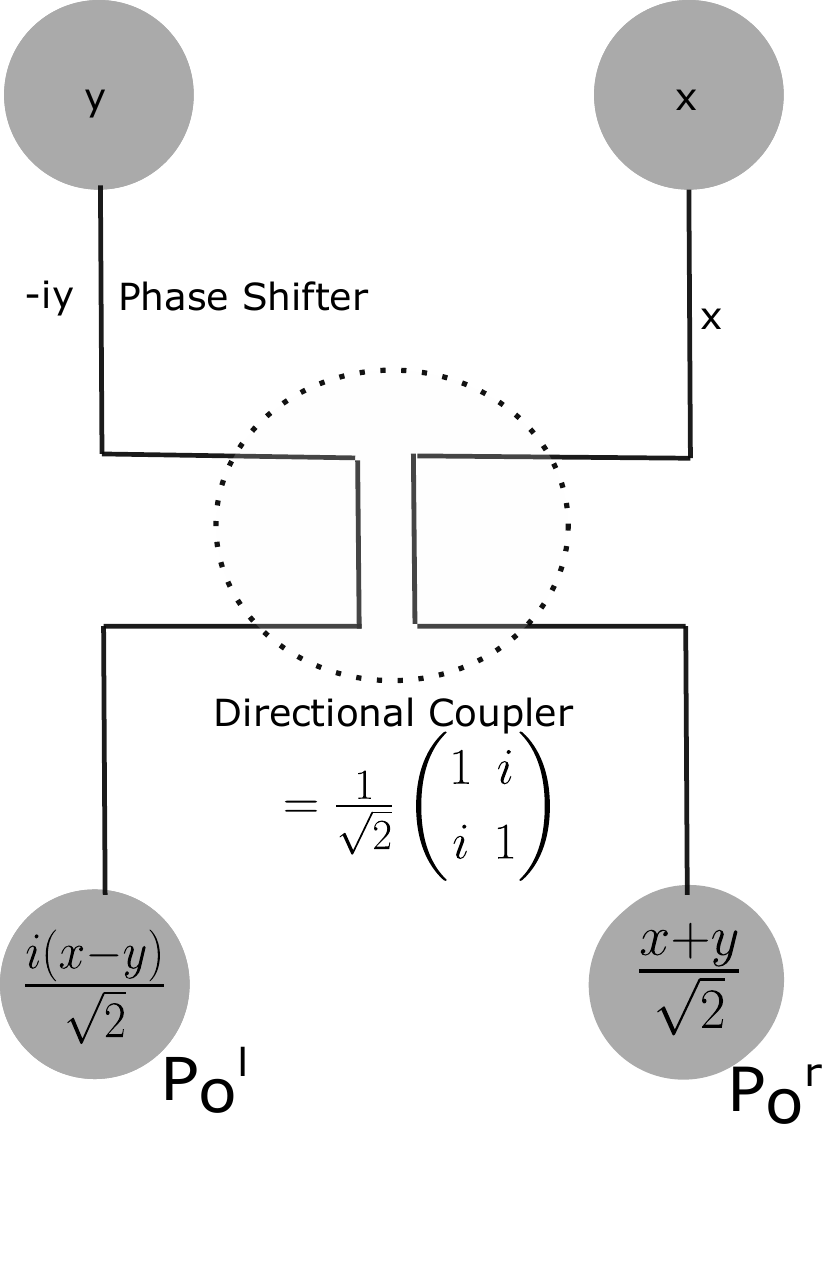}
\caption{Optical circuit}
\label{fig:optical_circuit}
\end{figure}

Let us lay down the following definitions.
Let us denote the starting state of the optical cavity by $\varphi$, where $\varphi $ is a column vector such that $\varphi = [x,y]^T$, where $x$ and $y$ are two numbers, encoded in the form of stimuli $J_1$ and $J_2$. The encoding is supported by Equation \ref{eqn:TE_mode_excitation_coefficient}, which shows that altering the magnitude and direction of the current source $J$ produces a direct effect on the generated magnetic and electric field distributions. Under the assumption that the data is in the range $[-1,1]$, the magnitude of the data is encoded as the amplitude of $J$, while the sign can be encoded in the form of the propagation axis of $J$. 

\subsubsection{Optical Set-Up}
The setup can be represented in terms of the starting state $\varphi= xu + yv = [u,v][x,y]^T,$ where $u$ is the column vector representing the right waveguide, which is $[1,0]^T$, $v$ is the column vector representing the left waveguide, which is $[0,1]^T$.  
Let us denote operations on the light state $\varphi$ through the matrix $U$ that denotes the gate operation: $U\varphi_1= \varphi_2,$ where $\varphi_1 $ is the initial state while $\varphi_2$ is the updated state after the gate operation. 
$U$ can be denoted as a 2 by 2 matrix, which is $\begin{pmatrix} a_{11} & a_{12}\\ a_{21} & a_{22} \end{pmatrix}$. Here, the gate has four ports, which are two input ports and two output ports. The entries $a_{11}, a_{12}, a_{21}, a_{22}$ map the input to the output ports. Let us say that the input ports carry $a$ and $b$, which can be represented by $ [a,b]^T$, while the output ports carry $a'$ and $b'$, which are represented by $ [a', b']^T$. $U$ maps $a$ and $b$ to $a'$ through $a'= a_{11}a + a_{12}b$ and to $b'$ through $b'= a_{21}a + a_{22}b$. Furthermore, $U$ is a unitary matrix, which means $U^\dagger U= I,$ where $I$ is the identity matrix. The unitary property of the gates ensures that $a$ and $b$ preserve their magnitude regardless of the operations performed (the gates only perform rotations and reflections). 
    
Let us lay down the definition of gate operations.

\begin{enumerate}
    \item Phase shifter $P(\theta)$: The gate preserves the phase of $a$ while rotating $b$ by $\theta$ radians anticlockwise. $P(\theta)= $ $\begin{pmatrix} 1 & 0\\ 0 & e^{i\theta}\end{pmatrix}$ 
    \item  Directional Coupler $H$: The gate takes in $a$ and $b$, splits them into two paths such that $a'$ receives $\frac{a+ib}{\sqrt{2}}$ and $b'$ receives $\frac{ia+b}{\sqrt{2}}$. $H=\frac {1}{\sqrt{2}}\begin{pmatrix} 1 & i \\ i & 1 \end{pmatrix}$
\end{enumerate}

The principle of the optical circuit (Figure \ref{fig:optical_circuit}) can be modeled by the equation below. 

\begin{equation}
\label{eqn:optical_circuit_equation}
    \begin{pmatrix} P_0^r \\ P_0^l \end{pmatrix}= HP(-\frac{\pi}{2}) \varphi = \frac {1}{\sqrt{2}}\begin{pmatrix} 1 & i \\ i & 1 \end{pmatrix} \begin{pmatrix} 1 & 0 \\ 0 & e^{-i\pi/2} \end{pmatrix} \begin{pmatrix} x \\ y\end{pmatrix}= \frac {1}{\sqrt{2}}\begin{pmatrix} x+y \\ i(x-y) \end{pmatrix}
\end{equation}


The photodiode (PD) at the end of each output port can convert incident signals into photocurrent (Figure \ref{fig:optical_cavity}). The generated photocurrent is proportional to the intensity of the signals, which is the square of optical magnitudes. Thus, the photocurrents generated on the right and left PDs denoted as $I_0$ and $I_1$ can be expressed as follows.

$ \begin{pmatrix} I_0 \\ I_1 \end{pmatrix} \propto \begin{pmatrix} ||x+y||^2 \\ ||i(x-y)||^2 \end{pmatrix} \propto  \begin{pmatrix} (x+y)^2 \\  (x-y)^2 \end{pmatrix}$,
The final output current is proportional to $((x+y)^2-(x-y)^2)$, which is proportional to $x.y$, which is the product of x and y.

\begin{stbox}
    The section presents the gate operations that are used to produce photocurrents that are proportional to $xy$, which is the product of two numbers $x$ and $y$ in the optical domain. These operations use directional couplers and phase shifters. Our goal is to reconstruct an optical cavity that is capable of performing these operations directly. Thus, we treat Equation \ref{eqn:optical_circuit_equation} as the target behavior in the reconstruction framework of the optical cavity. 
\end{stbox}

\begin{figure}[!htb]
\centering
\includegraphics[width=0.5\linewidth]{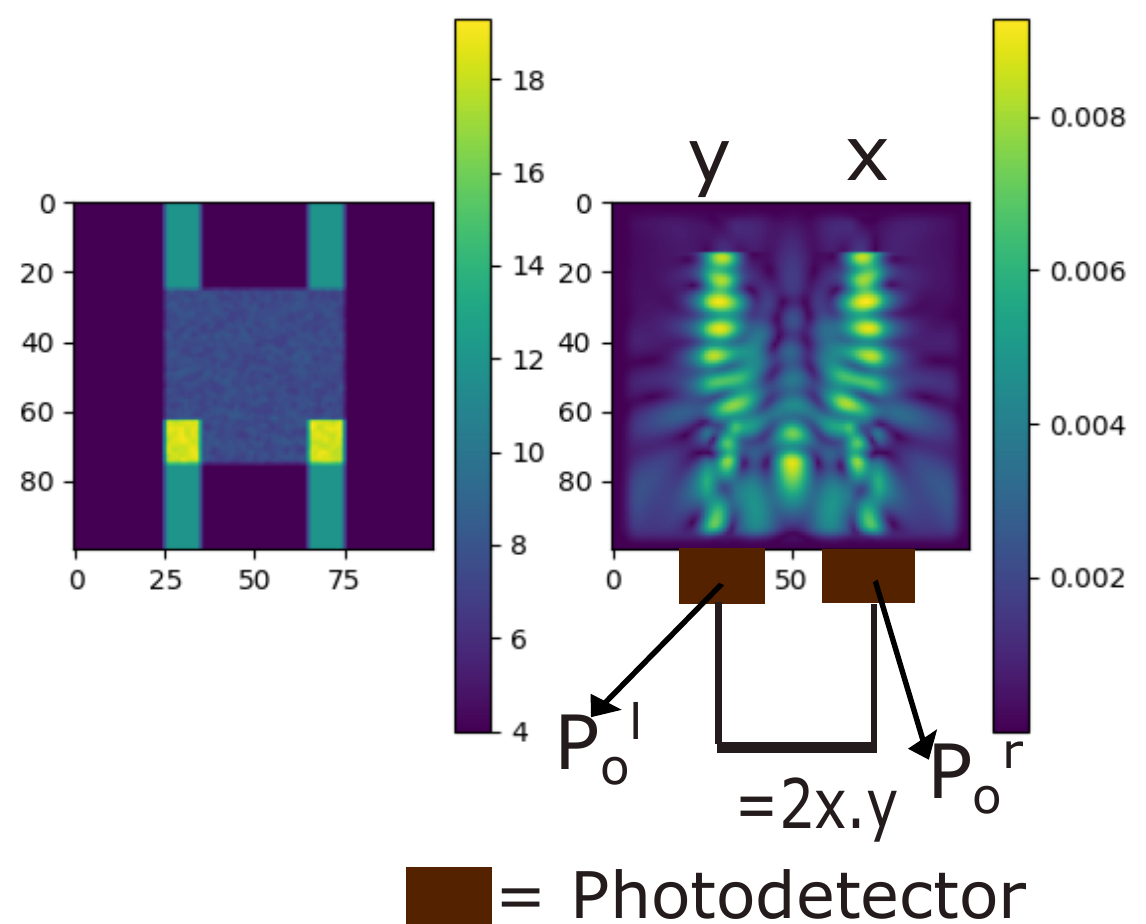}
\caption{Proposed optical cavity design}
\label{fig:optical_cavity}
\end{figure}

\subsection{Optical Cavity Setup} \label{sec:optical_cavity_setup}

\begin{figure}[H]
\centering
\includegraphics[width=0.5\linewidth]{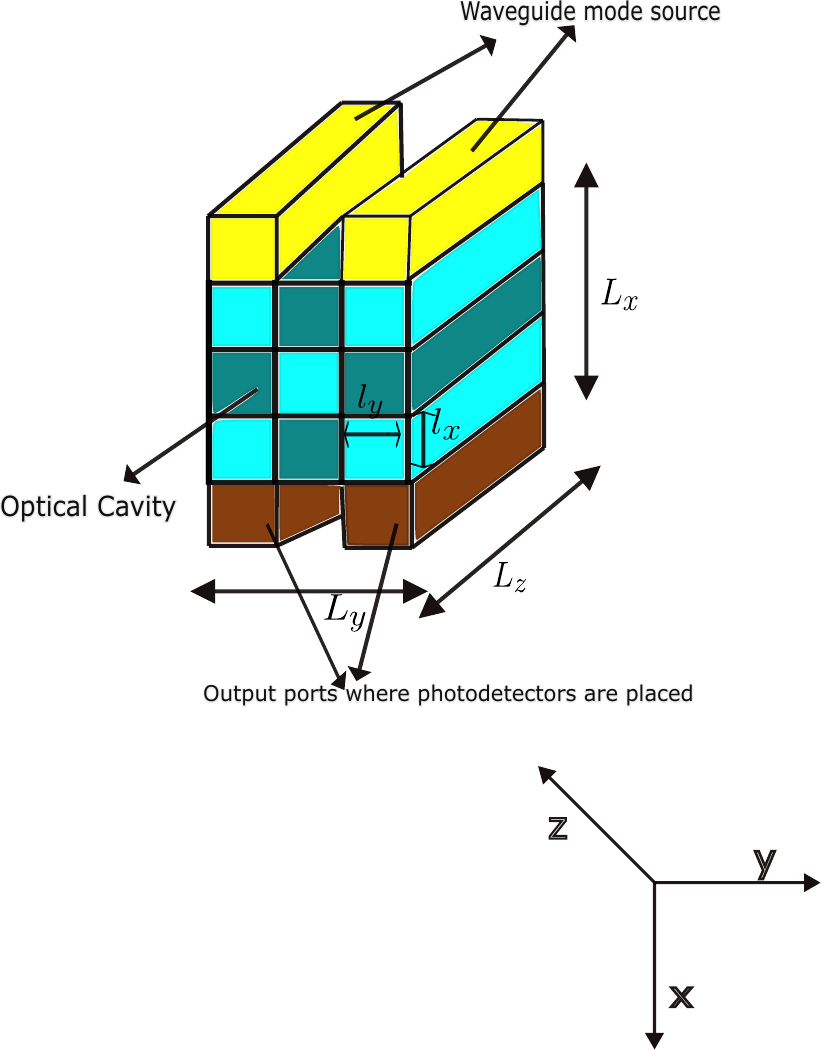}
\caption{3-D Cavity Matrix}
\label{fig:cavity_matrix}
\end{figure}

The waveguide mode source of Figure \ref{fig:cavity_matrix} is used to simulate light injection in the cavity according
to Equation \ref{eqn:TE_mode_excitation_coefficient}.    
 The cavity is represented as a 3-D matrix (Figure
\ref{fig:cavity_matrix}). 
 Suppose that we want to reconstruct an optical cavity of dimensions $(L_x, L_y, L_z)$ with
the mesh dimensions being $(l_x, l_y, L_z)$. The cavity is thus a 3-D matrix of dimensions $(\frac{L_x}{l_x},
\frac{L_y}{l_y},1 )$ where every entry decides the dielectric constant. 

If we restrict ourselves to two dielectric media $\epsilon_1, \epsilon_2 $, for every point in the grid i, $\epsilon_i=
\epsilon_1+p_i(\epsilon_2-\epsilon_1)$ is maintained, where $\epsilon$, which is the dielectric constant, is equal to $n^2$, where $n$ is the refractive index. Here, $p$ is the parameterization vector. If $p_i$ is $0$, $\epsilon_i$ becomes $\epsilon_1$. If $p_i$ is $1$, $\epsilon_i$ becomes $\epsilon_2$. However, assigning binary values to $p$ will mean adopting a discrete optimization-based approach. In order to be able to use powerful tools such as calculus mathematics, we let $p$ assume a continuous range of values in the range $[0,1]$. Thus, the parameter to be trained is $p \in  [0,1]^N,$ where N is the number of points in the grid\cite{vuck_invdes}. $\epsilon$ can be constructed from $p$ from the following equation.

\begin{equation}
    \epsilon= \epsilon_1 + (\epsilon_2-\epsilon_1)p
    \label{eqn:eps_from_p}
\end{equation}

\section{Optical Cavity Reconstruction Problem Statement}

In Equation \ref{eqn:optical_circuit_equation}, $HP(- \frac{\pi}{2})$ serves as the ideal gate operation expected to be performed by the reconstructed cavity. Let us represent the cavity in the form of a transfer function $M$ so that when two numbers, $x$ and $y$, encoded as current stimuli, are fed into the optical cavity, it produces the outputs $x'$ and $y'$, which represent the resultant electric field vector in the right and left output ports, respectively. It can be represented as follows.

\begin{equation}
    \label{eqn:optical_cavity_gate_operation}
   (x', y') = M(p,x,y) 
\end{equation}

Here, $x,y$ are two numbers in the range $[-1,1]$. $p \in [0,1]^N $ is the parameterization vector (Section \ref{sec:optical_cavity_setup}), which needs to be learnt from the measured data $x',y'$ (Equation \ref{eqn:optical_cavity_gate_operation}). Therefore, there are $N$ parameters to learn. They are reconstructed from the observed data $x',y' \in \mathbb{C}$. $M$ is an operator such that $[0,1]^N \times [-1,1]^2 \rightarrow \mathbb{C}^2$. 

Before being fed into the optical cavity, they are converted into the optical domain by using a source of current. The amplitude and direction of the current encode the magnitude and sign of a number, respectively (Equation \ref{eqn:TE_mode_excitation_coefficient}). The goal of the reconstruction is for $M$ to map the input $x,y$ to $x',y'$ such that $x'= \frac{1}{\sqrt{2}} (x+y)$ and $y'= \frac{1}{\sqrt{2}}i(x-y)$ serve as the target. 

Let us reconstruct $p$ by measuring $M$'s response to impulse signals as stimuli, which can be defined as Dirac delta functions (Equation \ref{eqn:dirac_delta}). This corresponds to the case where the inputs $x$ and $y$ are 1. The response to impulse signals will serve as the output of the cavity. 

The measure of fit, which determines how well $M$ fits the observed values with the target, can be defined as follows.

\begin{equation}
    \begin{split}
        f(p,M)= || (\frac{1}{\sqrt{2}},\frac{1}{\sqrt{2}}i)- M(p,1,0)||_2^2 + || (\frac{1}{\sqrt{2}},-\frac{1}{\sqrt{2}}i)- M(p,0,1)||_2^2
    \end{split}
    \label{eqn:measure_of_fit}
\end{equation}

Equation \ref{eqn:measure_of_fit} exploits the superposition principle of EM waves (Equation \ref{eqn:wave_eqn_inhomogeneous_media}) to treat the cavity outputs of two individual sources independently. Here, $||\vec{v}||_2$ is the $L_2$ norm of a complex-valued $\vec{v}$, which computes the square root of the summation of the squares of the absolute value of the individual components of the vector. $f$ is a $[0,1]^N \rightarrow \mathbb{R}_{\ge 0} $ mapping. 

Let us define the transformation $A$ such that it maps $p$ to the dielectric distribution $\epsilon$ (Equation \ref{eqn:eps_from_p}). $A$ is a $[0,1]^N \rightarrow [\epsilon_1, \epsilon_2]^N$ mapping. Let $B$ be a function that maps a number in the range $[-1,1]$ to a current stimulus $\vec{J}$ such that it is a mapping $[-1,1] \rightarrow \mathbb{C}^N$ (Equation \ref{eqn:efromj} shows that the matrices $\epsilon, E$ and $J$ should be of the same dimensions). Let $G$ be a function that maps $\vec{J}$ and $\epsilon$ to $\vec{E}$ using Equation \ref{eqn:efromj}. It is a $\mathbb{C}^N \times [\epsilon_1,\epsilon_2]^N \rightarrow \mathbb{C}^N$ mapping. Let us define a function $C_r$ and $C_l$ that map an electric field matrix generated by $G$ to a resultant vector in the right and left output ports, respectively. This is a $\mathbb{C}^N \rightarrow \mathbb{C}$ mapping and uses the superposition principle to superimpose the electric field vectors in the selected cross-section area of the output ports to generate a resultant vector. 

\begin{equation}
    \begin{split}
        \epsilon = A(p) \\
        J_1= B(x) \\
        J_2= B(y) \\
        E= G(J_1+J_2, \epsilon) \\
        x_1= C_r(E) \\
        x_2= C_l(E) 
    \end{split}
    \label{eqn:transformation}
\end{equation}

Let us define a unified function that maps $p, x, y$ to $x_1, x_2$ in Equation \ref{eqn:transformation}. Let us call it $O$, which is $O(p, x, y)=(C_r(G(B(x)+B(y),A(p))), C_l(G(B(x)+B(y),A(p))))$. It is a $[0,1]^N \times [-1,1]^2 \rightarrow \mathbb{C}^2$ mapping. 

Using Equation \ref{eqn:measure_of_fit} and $O$, let us mathematically define the reconstruction problem statement of $p$ along with the constraints on operator $M$. 

\begin{equation}
    \begin{split}
        min_p f(p,M), \\
        O(p,1,0) - M(p,1,0)= (0,0); \\
        O(p,0,1) - M(p,0,1)= (0,0)
    \end{split}
    \label{eqn:inverse_prob_statement}
\end{equation}

Let us define the cavity reconstruction problem as finding the distribution of $p \in \{0,1\}^N$ such that $f(p,M)$ of Equation \ref{eqn:inverse_prob_statement} becomes minimum. 

Let us prove that the problem statement is undecidable. The Halting Problem is a known $NP$ Hard problem. If given a machine $TM$ and input $w$, it decides whether $TM$ halts.

\begin{claim*}
    The Halting Problem can be reduced to the cavity reconstruction problem to show that the reconstruction is undecidable.
\end{claim*}

\begin{proof}
    Let us define a formula that serves as an input to the Halting problem and transform it into an equivalent input to Equation \ref{eqn:inverse_prob_statement}. Let us say that the transformation has $N$ variables $ \{p_i\}_{i=1}^N | p_i \in \{0,1\}$. We use the Turing machine $TM$ to generate all possible permutations of $ \{p_i\}_{i=1}^N$, which are of the order $2^N$. The machine accepts a configuration of $p$ and checks if the Halting problem halts for an input $p$. If it does halt, $f(p,M)$ is evaluated and the Turing Machine moves on to the next configuration. The process continues until the machine reaches the last configuration, which is the final end state; otherwise, it loops forever. If $p$ where $f(p,M)$ attains its lowest value is reported, the Halting problem would have halted. 
\end{proof}

The reconstruction problem is thus undecidable. The optimal solution can only be approximated. As established in earlier sections, we let $p$ assume a continuous range of values $\in [0,1]^N$ to use calculus-based tools to discover optimal regions of $p$.

\section {Objective Function F($\epsilon$)} \label{sec:obj_fn}

Let $f$ be a function that accepts two complex numbers $x$ and $y$ such that $f(x,y)= ||(x-y)||^2= (x-y)^\dagger(x-y)$, 
where $||a+ib||= \sqrt{a^2+b^2}$ is the absolute value of any complex number $a+ib$.   
Since z is a free direction in our cavity configuration, we will restrict the design to the x-y plane; that is, the normal is
$\hat{z}$. The complex plane can be interpreted as an x-y plane; the complex vectors can be interpreted as vectors where the
complex number $a+ib$ can be viewed as a vector $[a,b]^T$ where $a$ is the projection of the vector along the $x$ axis and $b$ is the projection of the vector along the $y$ axis.

Let $x$ represent the target, while $y$ is the predicted value. 
Let $F$ be the cost function that needs to be minimized. 

Suppose $F(x_1,x_2,...,x_n,y_1, y_2,...,y_n)= \sum_{i=1}^n||x_i-y_i||^2 = \sum_i f(x_i,y_i)$. Let us show that minimizing $F$ achieves the optimal condition where the predicted values align with the target. 

The $y_i,$ which corresponds to the field vectors, is a result of Equation \ref{eqn:efromj}. 

Let us rebuild $f(x_i,y_i)$ in terms of $\epsilon(p(\vec{r})),$ where $\epsilon(p(\vec{r}))$ is real. 

\begin{equation}
    \vec{E}_i(\vec{r})=((\nabla \times \frac{1}{\mu}\nabla \times)- \omega^2\epsilon_0\epsilon(p(\vec{r}))^{-1}(i\omega J_i (\vec{r}))
    \label{eqn:gen_e}
\end{equation}

Here $J_i$ is the current density induced by the source $i$ and is treated as a 3-D matrix, $E_i$ is the electric field induced by the source $i$ and is represented as a 3-D matrix, and $\epsilon$ is the dielectric distribution, represented as a 3-D matrix. The matrices are of dimension $(N_x, N_y, 1)$, where $N_x$ and $N_y$ are the number of grid points along the $x$ and $y$ axes, respectively. 

We shall optimize the optical cavity with respect to the fundamental mode of propagation (the mode with the lowest cutoff frequency to exist). Let us define $c$ to denote the extent of the measured field overlap with the fundamental mode. The following set of equations measures it. Let us denote the measured electric and magnetic fields by $\vec{E}$ and $\vec{H}$, respectively. Let us denote the fundamental electric and magnetic fields by $\vec{E}_{mode}$ and $\vec{H}_{mode}$, respectively. Let $\hat{n}$ be the normal vector of the cross-sectional area $S$. Let $o$ denote the overlap integral. 

\begin{equation}
     o= \int_S ( \vec{E} \times \vec{H}_{mode} + \vec{E}_{mode} \times \vec{H}). \hat{n} dS
\end{equation}

$\vec{H}$ can be written in terms of $\vec{E}$ using Equation \ref{eqn:eh}.

\begin{equation}
    \begin{split}
         & \vec{H}(\vec{r}) = \frac{1}{i\omega \mu_0 } \nabla \times \vec{E}(\vec{r}) \\
        \therefore o & = \int_S \left (\vec{E} \times \vec{H}_{mode}+  \frac{1}{i\omega \mu_0 } \vec{E}_{mode} \times (\nabla \times \vec{E}) \right).\hat{n} dS \\
        & = \int_S \left (\vec{E} \times \vec{H}_{mode} - \frac{1}{i\omega \mu_0 } (\nabla \times \vec{E}) \times \vec{E}_{mode} \right).\hat{n} dS \\
        & = \int_S \left (\vec{E} \times \vec{H}_{mode} + \frac{1}{i\omega \mu_0 } (\vec{E} \times \nabla) \times \vec{E}_{mode} \right).\hat{n} dS \\
        &= \int_S \vec{E} \times \left (\vec{H}_{mode} + \frac{1}{i\omega \mu_0 } \nabla \times \vec{E}_{mode} \right).\hat{n} dS \\
        &= \int_S \vec{E} . \left ( \left (\vec{H}_{mode} + \frac{1}{i\omega \mu_0 } \nabla \times \vec{E}_{mode} \right) \times \hat{n} dS \right )  
    \end{split}
\end{equation}

Let us define $c$ in terms of an overlap vector. 

\begin{equation}
    \vec{c}= \left ( \left (\vec{H}_{mode} + \frac{1}{i\omega \mu_0 } \nabla \times \vec{E}_{mode} \right) \times \vec{S} \right )
    \label{eqn:overlap_vector_estimation}
\end{equation}

In the world of computation, $\vec{c}$ is treated as a 3-D matrix ($\vec{r}$ is simply a 3-D matrix representing the spatial distribution of points in the grid), which we call $c$. Using Equation \ref{eqn:overlap_vector_estimation}, $c^{\dagger} E$, where $E$ is the 3-D matrix representing $\vec{E}$, will produce a 3-D matrix that represents the measure of how much the measured fields align with the desired mode. Let us define a function $C$ that takes a matrix $M$ with complex entries as input and converts it into a complex number by adding all the entries of $M$. $C(c^{\dagger}E)$ generates a complex number of form $a+ib$ that, as previously established, can be treated as a vector $[a,b]^T$. The transformation holds because of the superposition property of EM waves. The matrix $c^{\dagger}E$ contains the distribution of the electric fields corresponding to the fundamental mode of propagation within the desired cross-sectional area. The resultant field is simply the sum of the fields spread throughout space, which is what the transformation $C$ does. $||C(c^{\dagger}E)||^2$ measures the power contained within the desired mode \cite{vuckovic_spins_arch}. 

$f$ becomes the following. 

\begin{equation*}
\begin{split}
    f(C(c^\dagger E_i), x_i) &= || C(c^\dagger E_i)- x_i||^2\\
    &=( C(c^\dagger E_i(\epsilon))- x_i)^\dagger( C(c^\dagger E_i(\epsilon))- x_i)
\end{split}
\label{eqn:f_defn}
\end{equation*}

\begin{figure}[H]
\centering
\includegraphics[width=0.5\linewidth]{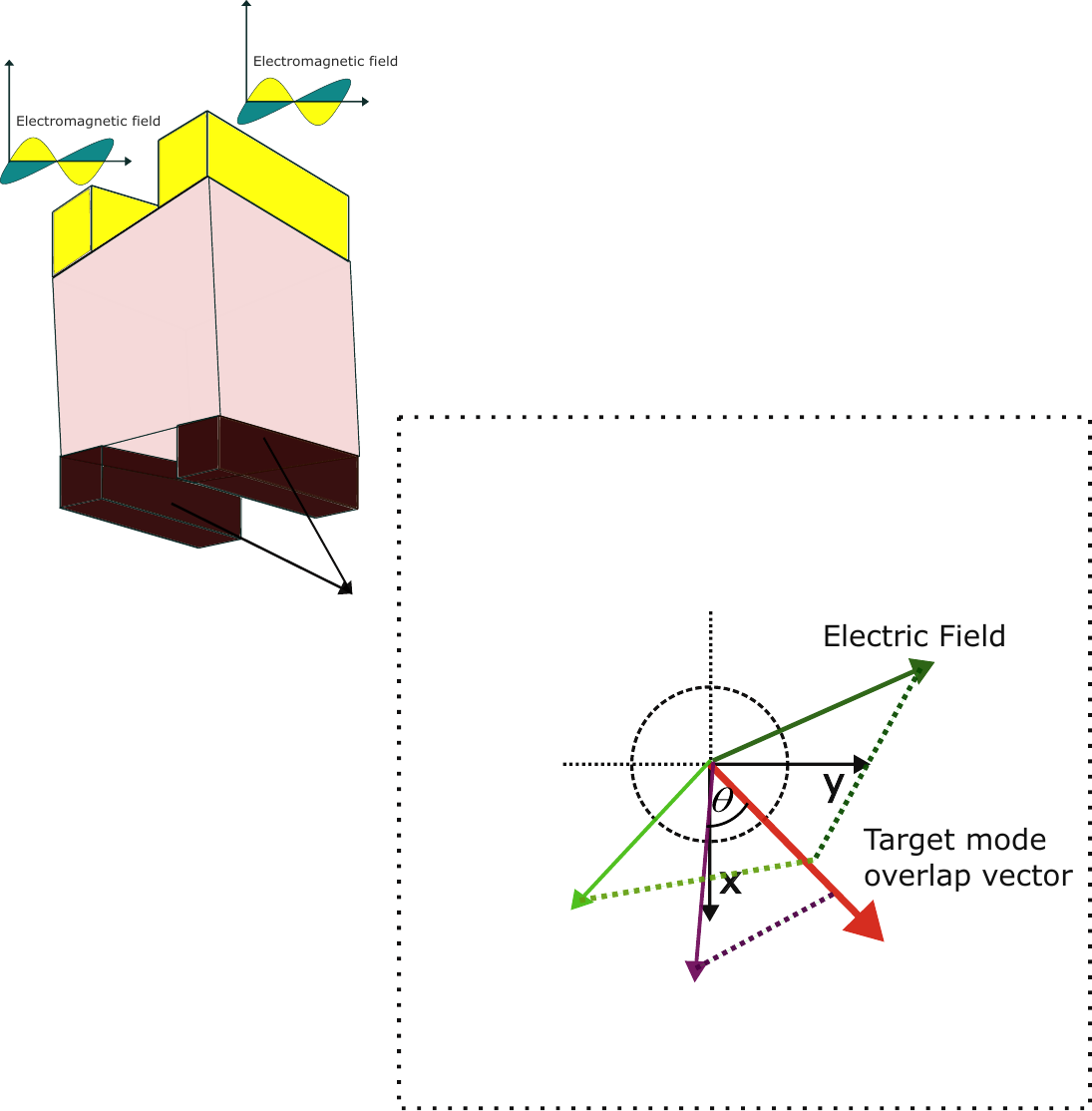}
\caption{Measuring electric field overlap with the target mode}
\label{fig:mode_overlap_measurement}
\end{figure}

In Figure \ref{fig:mode_overlap_measurement},  $\theta$ is the angle between the overlap vector for the target mode of propagation and the $x$ axis. Here, the target mode is the fundamental mode. The overlap vector can be calculated by
Equation \ref{eqn:overlap_vector_estimation}. The system involves two sources of light and two ports where photodetectors are placed
to measure the intensity of electric fields transmitted by the optical cavity. Let us call $\theta_{11}$ as the angle that the 
overlap vector makes with the $x-$ axis at the output port 1 (right output port in Figure \ref{fig:cavity_matrix}) when source 1 
produces the fields in the cavity; $\theta_{12}$ is the angle that the overlap vector makes with the $y-$ axis at the output port 2 (left output port in Figure \ref{fig:cavity_matrix})
when source 1 is the cause; $\theta_{21}$ is the angle that the overlap vector makes with the $x-$ axis at the output port 1 when source 2 is the cause; $\theta_{22}$ is the angle that the overlap vector makes
with the $y-$ axis at the output port 2 when source 2 is the cause.  

Mimicking the optical circuit principle in Equation \ref{eqn:optical_circuit_equation} in the domain of optical cavity, the
addition term "$x+y$"  can be replaced by constructive interference of electric fields generated by the two sources, encoding $x$ 
and $y$, at the right output port (output port 1) while the subtraction term "$x-y$" can be replaced by destructive interference
at the left output port. The objective function $F$ becomes the following. Let us say $\vec{\theta}$ is $[\theta_{11}, \theta_{12}, \theta_{21}, \theta_{22}]$.  

\begin{equation}
\begin{split}
     F_{\vec{\theta}}(\epsilon(p))=  f(C(c_{r}^{\dagger }E_1), \frac{1}{\sqrt{2}}e^{i\theta_{11}})+f(C(c_{r}^{\dagger }E_2), \frac{1}{\sqrt{2}}e^{i\theta_{21}})   + \\
f(C(c_{l}^{\dagger }E_1), \frac{1}{\sqrt{2}}e^{i\theta_{12}})+f(C(c_{l}^{\dagger }E_2), \frac{1}{\sqrt{2}}e^{i\theta_{22}})
\end{split}
\label{eqn:obj_function}
\end{equation}

Here, $\dagger$ is the conjugate transpose. 

Equation \ref{eqn:wave_eqn_inhomogeneous_media} shows that the resultant electric field due to multiple sources of light in the system can be treated as a summation of electric fields generated by the individual sources (since the superposition principle holds). This result is used to establish the objective function in Equation \ref{eqn:obj_function}. 

In Equation \ref{eqn:obj_function}, the term $\frac{1}{\sqrt{2}}$ appears to divide the power carried by the fields equally among the two output ports. The parameter $\vec{\theta}$ controls the phase of the electric fields: the difference $(\theta_{11}- \theta_{21})$ should be kept equal to even multiples of $\pi$ to steer constructive interference of $E_1$ and $E_2$ at the right output port. The difference $\theta_{12}-\theta_{22}$ should be kept equal to odd multiples of $\pi$ to cause destructive interference of $E_1$ and $E_2$ in the left output port. $c_l$ is the overlap vector for the left output port plane, where the normal of the plane is the $y-$ axis; $c_r$ is the overlap vector for the right output port plane, where the normal of the plane is the $x-$ axis. $p$ is the parameterization vector defined in Section \ref{sec:optical_cavity_setup}. 

\begin{figure}[H]
    \begin{subfigure}{0.3\textwidth}
        \includegraphics[width=\linewidth]{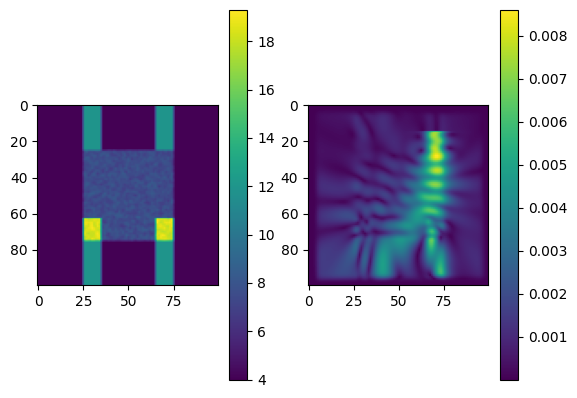}
        \caption{source-1}
    \end{subfigure}
    \hfill
    \begin{subfigure}{0.3\textwidth}
        \includegraphics[width=\linewidth]{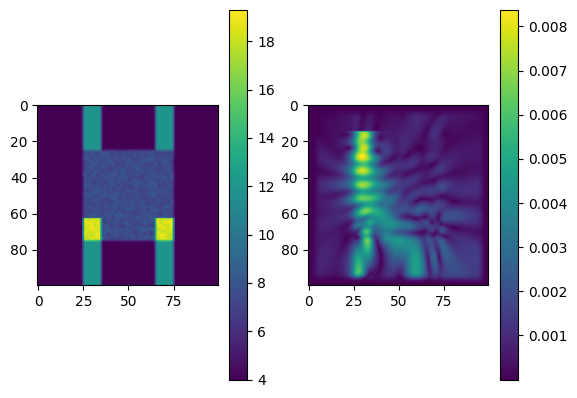}
        \caption{source-2}
    \end{subfigure}
  
    \caption{(a) source-1 on (b) source-2 on}
    \label{initial set-up}
\end{figure}

Figures \ref{initial set-up} a and b show electric fields induced by sources 1 and 2, respectively, which are used to
calculate the objective function $F$ in Equation \ref{eqn:obj_function}.  

Suppose $\vec{\theta}$ has been fixed (the target vectors have been determined), we denote $F_{\vec{\theta}}(\epsilon(p))$ by $F$ or $F(p)$ or $F(\epsilon)$.

Let us denote $C(c_j^{\dagger}E_i)$ by $C_{ij}$ and the corresponding target vectors by $x_{i,j}$, 
where $j$ $\in \{ r,l \}$ ($r$ stands for the right output port and $l$ for the left output port)  and $i \in \{1,2 \}$ (where $1$ stands for source-1 and $2$ for source-2). 

The cost function becomes the following.
\begin{equation}
    F(\epsilon(p))= \sum_{i\in{1,2}, j\in{r,l}}(C_{ij}-x_{i,j})^{\dagger}(C_{ij}-x_{i,j})
    \label{eqn:F_in_terms_of_overlap_vec}
\end{equation}

\begin{equation}
\begin{split}
    \frac{dF}{dp}= \sum_{i,j} ((C_{ij}-x_{i,j})^\dagger\frac{\partial C_{ij}}{\partial \epsilon}\frac{\partial \epsilon}{\partial p} + 
    (C_{ij}-x_{i,j})\frac{\partial C_{ij}^\dagger}{\partial \epsilon}\frac{\partial \epsilon}{\partial p}) 
\end{split}
    \label{eqn:cost_fn_gradient}
\end{equation}

$F$ achieves the lowest value when $\frac{dF}{dp}$ becomes $0$: we know that there is a minimum value for $F$
since $F$ is simply a sum of squares of residual errors (difference between predicted and target values); therefore, when the derivative of $F$ with respect to $\epsilon$ (design parameters to optimize) reaches $0$, it indicates that there is no scope for further improvement ($F$ ceases to change at that point). Section \ref{sec:gradient_based_opt_route} elaborates and discusses the gradient-based optimization route mentioned above.  


Equation \ref{eqn:obj_function} prompts the need to make wise selections of the phase term $\vec{\theta}$ to fix the alignment of the target electric fields in the output ports, where the photodetectors must be placed to measure the resulting optical magnitude. The following constraints apply to phases: $\theta_{21}= \theta_{11}$ (to cause constructive interference in the right output port) and $\theta_{12}= \theta_{22} \pm \pi$ (to cause destructive interference in the left output port). $\theta_{11}$ and $\theta_{12}$ can assume any values in the range $[0,2\pi]$. $\theta_{21}$ and $\theta_{22}$ will assume values according to the above constraints. We adopt a hill-climbing-based approach to make such decisions. 

We shall denote $\epsilon(p)$ by $\epsilon$ in later sections. 

\begin{stbox}
    The section establishes the objective function for our reconstruction problem. It defines a transformation $C$ that computes the resultant electric field vector contained within the defined cross-sectional area of the output port. Equation \ref{eqn:obj_function} uses the superposition principle established in Equation \ref{eqn:wave_eqn_inhomogeneous_media} that lets us treat EM waves from independent sources independently. $F_{\vec{\theta}}(\epsilon(p))$ depends on the phase information. The fundamental mode profiles of electric fields can assume any phase with respect to the $x$ or $y$ axis in the range $[0, 2\pi]$. Let us say $\theta_{ij}$ refers to the phase of the resulting electric field vector, calculated by $C$, from the source $i$ in the output port $j$. We treat $\theta_{11}$ and $\theta_{12}$ as tunable knobs that are free to accept any value in the range $[0, 2\pi]$, while $\theta_{21}$ and $\theta_{22}$ depend on $\theta_{11}$ and $\theta_{12}$, respectively, to steer constructive and destructive interferences in the right and left output ports. Once the phase information has been determined, the objective function $F$ can be evaluated to learn the dielectric distribution.  
    \end{stbox}
\section{Objective Function Optimization}


\subsection{Selecting $\vec{\theta}$: Hill Climbing-Based Search} \label{sec: hill climbing for phase selection}

Let us define an operator $H$ such that it takes $\vec{\theta}$ as input and produces the reconstructed $\epsilon$ along with a measure of fit, determined by $F_{\vec{\theta}}(\epsilon)$ of Equation \ref{eqn:obj_function}. The operator initializes $p$ with random values in the range $[0.3,0.7]$ and therefore does not need $p$ as input. It can be considered as a function $H: [0,2\pi]^4 \rightarrow [\epsilon_1,\epsilon_2]^N \times \mathbb{R}_{\ge 0}$. 

\begin{stbox}
    \textbf{Problem Statement}  \\
    Input: A function $H(\vec{\theta})$ that uses $\vec{\theta}$ to establish the target vectors in Equation \ref{eqn:obj_function} and generates the reconstructed $\epsilon$ that minimizes $F_{\theta}(\epsilon)$ along with the corresponding value of $F_{\theta}(\epsilon)$. \\
    Task: Find a $\vec{\theta} \in [0, 2\pi]^4$ that minimizes $F_{\theta}(\epsilon)$ produced by $H(\vec{\theta})$ and return the corresponding $\epsilon$. 
\end{stbox}

Let us assign the following definitions. Let $N$ be an arbitrary number belonging to the set of natural numbers and let $\delta$ be an arbitrary real number in the range $(0,1)$.  $\theta_{11}$ and $\theta_{12}$ are the two tunable knobs that can assume any real number in the range $[0, 2\pi ]$. Let us consider the $2\pi$ range as a circle centered at a point $C$; dividing the circle into $N$ arcs will result in each arc making an angle of $\frac{2\pi}{N}$ at the point $C$ (Figure \ref{fig:phase_search_circle}).   

\begin{figure}[!htb]
\centering
\includegraphics[width=0.8\linewidth]{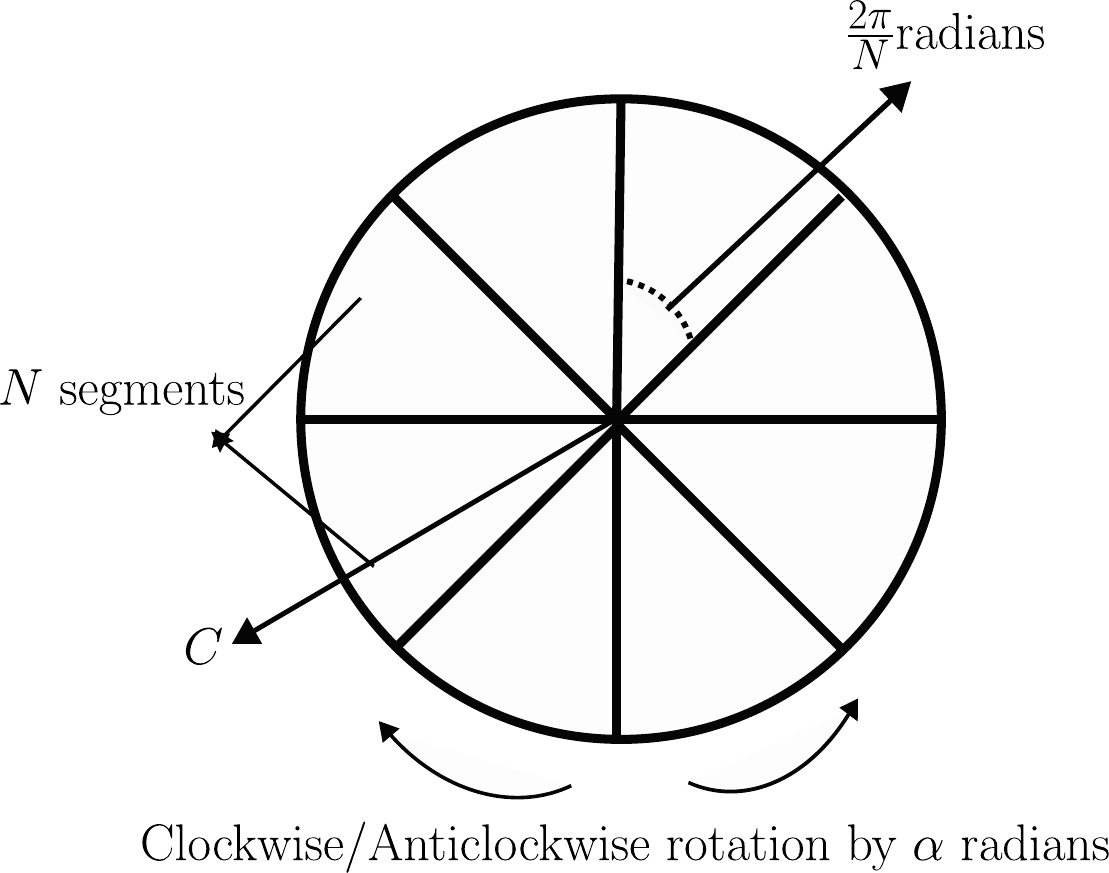}
\caption{Dividing $2\pi$ into $N$ segments}
\label{fig:phase_search_circle}
\end{figure}

In Figure \ref{fig:phase_search_circle}, $\alpha$ denotes the angle in radians that the circle must be rotated
clockwise/anticlockwise to reach some other point on the circle. $N$ segments constitute $2N$ points on the
circumference of the circle to explore. Let us call these points $ \{P_i\}_{i=1}^{2N}$. For every $P_i$, we launch a
hill-climbing-based search to look for nearby promising points by rotating the circle anticlockwise or clockwise by
$\alpha$. This parallelizes the search space, since the $P_i$s can be treated as independent starting points for hill
climbing that can be launched on multiple cores simultaneously. Moreover, it gives the hill climbing algorithm multiple starting points to avoid the issue of the search being stuck at a local minimum. 

\begin{algorithm}[!htb]
    \caption{Hill Climbing Based Search} \label{alg: hill climbing}
    \begin{algorithmic}
        \Require Function $H$; $N$ segments; $\alpha$; maximum iterations $M$; tolerance $\delta$; objective function $F_{\vec{\theta}}(p)$
        \Ensure For every segment, if the sequence of executed states, where a state is $[\epsilon,\vec{\theta}]$, at $j^{th}$ step is $\{q_i\}_{i=0}^{j},$ $F_{q_i[1]}(q_i[0]) < F_{q_{i-1}[1]}(q_{i-1}[0]) \forall i\ge 1$.
        
        \State $j \gets 0$ 
        \While {$j <  N$} \Comment{The loop can be executed in parallel}
        \State $\theta_{11} \gets \frac{2\pi}{N}j$
        \State $\theta_{22} \gets \frac{2\pi}{N}j$
        \State $\theta_{12} \gets \theta_{22}- \pi $
        \State $\theta_{21} \gets \theta_{11}$
        
        \State $i \gets 0$
       
        \State $\vec{\theta} \gets [\theta_{11}, \theta_{12}, \theta_{21}, \theta_{22}]$

        \State $\epsilon, \texttt{obj} \gets H(\vec{\theta})$
        \State $q \gets [\epsilon, \vec{\theta}]$

        \While{$i < M$}

         \State \text{Toss a coin}
         \If{\text{The coin shows heads}}
         \State $\theta_{11} \gets \theta_{11}+ \alpha$
         \Else 
         \State $\theta_{11} \gets \theta_{11} - \alpha$
         \EndIf

          \State \text{Toss the coin again}
         \If{\text{The coin shows heads}}
         \State $\theta_{12} \gets \theta_{12}+ \alpha$
         \Else 
         \State $\theta_{12} \gets \theta_{12} - \alpha$
         \EndIf

         \State $\theta_{21} \gets \theta_{11}$
         \State $\theta_{22} \gets \theta_{12} - \pi$

         \State $\vec{\theta} \gets [\theta_{11}, \theta_{12}, \theta_{21}, \theta_{22}]$

         \State $p, \texttt{obj}' \gets H(\vec{\theta})$
         \State $q' \gets [\epsilon, \vec{\theta}]$

         \If{$ absolute(\texttt{obj}- \texttt{obj}') < \delta $} \Comment{Check if the search has converged}
         \State \text{Convergence is reached}
         \State $break;$
         \EndIf
         
         \If{$ \texttt{obj}' < \texttt{obj}$}
         \State $q \gets q'$ \Comment{Accept this state}
         
         \EndIf
         
        \State $i \gets i+1$
        
        \EndWhile

        \State $j \gets j+1$
        \EndWhile

    \Return \text{state $q$ which yields least $F_{q[1]}(q[0])$}
    
    \end{algorithmic}
\end{algorithm}

In Algorithm \ref{alg: hill climbing}, $N$ is assigned to $48$ as we had 48 cpu cores, $\alpha$ to 0.01 as we did
not want to keep the precision of $\theta$ beyond two decimal places, $M$ to $3$ (the total number of iterations is
$M*N$, which is $144$) and $\delta$ to $10^{-6}$. $2\pi \approx 6.28$. Since $\alpha$ remains at $0.01$, the total number of
points to explore is $628$. $\frac{2\pi}{48}$ is $0.13$ radians. Each of the $48$ segments covers an angle of $0.13 $ radians.
Therefore, the angles are covered at an interval of $0.13$ radians. Furthermore, at every point $P_i,$ the hill climbing
search can rotate the circle in figure \ref{fig:phase_search_circle} by $\alpha$, which is $ 0.01$ radians, anti-clockwise or
clockwise, giving us $\frac{0.13}{\alpha}$, which is 13 points to explore on either side of $P_i$.  The optimal point (or local
minimum), when $P_i$ is the starting point, must lie on one of the sides of the starting point. At the starting point,
$3$ attempts are made to explore the neighbourhood($M=3$). The probability that an attempt traverses the side where the
optimal solution lies (assuming that the starting point does not lie on a local minimum) is $\frac{1}{2}$. One of the $13$
points is the optimal point in the neighborhood of the starting point. The probability that the optimal point is indeed
reached in $M$ (which is equal to 3) iterations is equivalent to the probability that the optimal neighbor is among the top $M$ (which is 3) nearest neighbours of
$P_i$, which is P(optimal point is the nearest neighbour) + P(optimal point is the $2^{nd}$ nearest neighbour) + P(optimal point
is the $3^{rd}$ nearest neighbour) $= (\frac{1}{2})\frac{1.12!}{13!} + (\frac{1}{2})^2\frac{1.12!}{13!} +
(\frac{1}{2})^3\frac{1.12!}{13!} = 0.067$. 

\begin{stbox}
    The section presents a hill-climbing based selection approach to determine phase profile. In the interest of time, the process first divides the search space covering $2\pi$ radians into $2N$ segments, where $N$ is the number of available cores, which can be used to launch simulations simultaneously. Given that $N$ is 48, the angles at an interval of $0.13$ radians are covered. Furthermore, at each of the $2N$ points, hill-climbing is used as a local search to cover angles at an interval of $0.01$ radians. The number of iterations for the local search was kept low. The probability of the search finding the optimal point in the radius of $0.13$ radians with values precise up to the second decimal point is $0.067$. 
\end{stbox}

\subsection{Dielectric Distribution Reconstruction} 

Let us establish the function $H$ that is used to construct $\epsilon$ from a given $\vec{\theta}$ using the objective function of Equation \ref{eqn:obj_function}. Here, $\vec{\theta}$ is a known quantity. In order to mathematically define the construction of $H$, we define two functions $X: [0, 2\pi]^4 \rightarrow [0,1]^N \times \mathbb{R}_{\ge 0}$ and $Y: [0,1]^N \times \mathbb{R}_{\ge 0} \rightarrow [\epsilon_1, \epsilon_2]^N \times \mathbb{R}_{\ge 0}$. $X$ takes $\vec{\theta}$ as input, initializes $p$ with random values $\in [0.3,0.7]^N$ and learns $p$ through the gradient of $F_{\vec{\theta}}(p)$ with respect to $p$ (Section \ref{sec:gradient_based_opt_route}). $X$ returns the reconstructed $p$ and the corresponding value of $F$. Using $p$ returned by $X$, $Y$ converts $p$ to $\epsilon(p)$ and learns the grid indices where symmetry can be injected (Section \ref{sec:search_tree}) and returns perturbed $\epsilon$ if the value of $F(\epsilon)$ improves or else returns the $\epsilon(p)$ constructed by $X$ along with the corresponding $F(\epsilon)$. The role of $Y$ is to explore the effect of symmetry on the trapping of light. In terms of $X$ and $Y$, $H$ is $ YX(\theta) = H(\theta)$. 

\section{Gradient-Based Learning} \label{sec:gradient_based_opt_route}

In Section \ref{sec:optical_cavity_setup}, the definition of the parameterization vector $p$ was established, which controls the distribution of refractive indices in the cavity. 

The principal idea behind the reconstruction is to update $p$ through $p : p- \alpha * \frac{dF}{dp}$, where $\alpha$ is the learning rate, and $F$ is the objective function of Equation \ref{eqn:obj_function}. 

Let us recall Equation \ref{eqn:cost_fn_gradient}. In the equation, $\frac{\partial C_{ij}}{\partial \epsilon}$ is $C(c_j^\dagger \frac{\partial E_i}{\partial \epsilon})$.

Let us understand how $ \frac{\partial E_i}{\partial \epsilon}$ is computed, which involves the adjoint method \cite{adjmeth}. 

Let $g(\epsilon)$ be $((\nabla \times \frac{1}{\mu}\nabla \times)- \omega^2\epsilon_0\epsilon) E + (-i\omega J) = 0$ 
be the function that solves $E$ matrix for an epsilon distribution $\epsilon$ and current $J$. Let $f(\epsilon)$ be the objective function that represents the function $f$ in Equation~\ref{eqn:f_defn}, which depends on $\epsilon$ as the $E_i$ matrices in the equation are generated from the $\epsilon$ matrix (see Equation \ref{eqn:gen_e}). 

Since $g(\epsilon)$ is $0$, therefore $d_\epsilon g$ is $0$. 

 \begin{equation*}
     \begin{split}
           \implies g_EE_\epsilon +g_\epsilon = 0 \\ 
           \therefore E_\epsilon = -g_E^{-1}g_\epsilon \\ 
           d_\epsilon f = \partial_Ef\partial_\epsilon E \\
          =f_EE_\epsilon \\ 
           \therefore d_\epsilon f = -f_Eg_E^{-1}g_\epsilon
     \end{split}
 \end{equation*}

Let us define an equation that connects $g_E$ with $f_E$. Let us define $\lambda$ such that $g_E^T\lambda=-f_E^T$. This implies $\lambda^T= -f_E g_E^{-1}$. Here, $\lambda$ is called a matrix of adjoint variables. $d_\epsilon f$ becomes $\lambda ^Tg_\epsilon$. This is called the adjoint equation.


Here, $g$ simulates the fields that flow through the dielectric distribution. The question now arises as to how perturbations
in the dielectric media $\partial \epsilon$ affect g, which, in turn, affects the flow of the field. Let us see how $\frac{df}{dp}$ is calculated \cite{giles2000introduction}. 

\begin{equation*}
    \begin{split}
        g_E= ((\nabla \times \frac{1}{\mu}\nabla \times)- \omega^2\epsilon(p))^T 
    \end{split}
\end{equation*}

\begin{equation}
    ((\nabla \times \frac{1}{\mu}\nabla \times)- \omega^2\epsilon_0\epsilon(p)) (-i\omega \lambda) = i\omega f_E^T
\end{equation}

$f_E^T$ can be treated as the source $\vec{J}$ in Equation \ref{eqn:efromj}, while the calculated fields $\vec{E}$ correspond to $-i\omega\lambda$. 

\begin{equation*}
    g_\epsilon= -\omega^2 \epsilon_0 E 
\end{equation*}

\begin{equation}
    \begin{split}
        \implies d_\epsilon f &= \lambda^T (-\omega^2\epsilon_0 E) \\
         &= -\omega^2 \epsilon_0\lambda^T E
    \end{split}
\end{equation}

$d_pf= d_\epsilon f d_p \epsilon$ can thus be computed from $d_{\epsilon}f$, which is not more difficult to calculate than solving Equation \ref{eqn:efromj}. 

The problem involves reconstructing the $\epsilon$ distribution ($p$ vector) such that it can generate $E$ through $g$ that minimizes
$F$. Essentially, we need to find $\epsilon(p_{opt})$ that serves as a global minimizer of $F(\epsilon(p))$. The key idea behind
finding the minimizer is Newton's method \cite{mit_ocw_newtons_method}. Suppose that we start with an initial guess $p_0$. Let us
use Taylor's series to approximate $F(\epsilon(p))$ around the point $p= p_0.$ 

\begin{equation}
\begin{split}
    F(\epsilon(p))= F(\epsilon(p_0)) + \nabla F(\epsilon(p_0))^T (p - p_0) \\
    + \frac{1}{2}(p-p_0)^T \nabla^2 F(\epsilon(p_0)) (p-p_0) 
\end{split}
    \label{eqn:Taylor_approx_of_F}
\end{equation}

$\nabla^2 F$ is the Hessian of $F$ and can be denoted by $H$. Let us define $\Delta p= p-p_0$. Let us analyze the gradient of $F$ using Equation \ref{eqn:Taylor_approx_of_F}. $\nabla F(\epsilon(p))= \nabla F(\epsilon(p_0)) + H(p_0) (p-p_0) $. Plugging
$\nabla F =0$ will fetch the direction of minimization: $p= p_0 - \frac{\nabla F(\epsilon(p_0))}{H(p_0)}$. $\Delta p= p-p_0= -\frac{\nabla F(\epsilon(p_0))}{H(p_0)}$ is called the Newton direction. Large problems that involve a large number of variables to optimize mean that the Hessian matrices become difficult or bulky to compute. L-BFGS-B algorithm is recommended to handle bulky optimization problems, such as in our case \cite{lbfgs}.  

An important part of the reconstruction now becomes choosing the direction along which $p$ must be updated (keeping in mind that the search space can involve local minima and saddle points). Suppose $\alpha$ represents the parameter
that controls the rate of change of $F$ along $p$: $\Delta p '= \alpha \Delta p, $ where $\alpha \in [0,1]$. 

A symmetric matrix with positive eigenvalues is said to be positive definite. If the Hessian of $F$, $\nabla^2 F,$ is positive definite, Newton direction $\Delta p$ will eventually lead to the global minima of $F$. The explanation is as follows. Let $A$ be a square matrix. Suppose $A$ is positive definite, which means that $A$ must be symmetric: $A= A^T$. $\nabla ^2 F$ (or $H$) is ensured to be a symmetric matrix. Moreover, a positive definite matrix has positive eigenvalues (a symmetric matrix has real-valued eigenvectors): for any $x$ belonging to the class of real-valued vectors, $x^\dagger A x >0$. Assuming that $H$ also fulfills this condition, $F$ is a convex function \cite{cs229_cvxopt_convex}. $\Delta p^\dagger H(p_0) \Delta p = - \nabla F(\epsilon(p_0))^T \Delta p > 0 $, which implies $ \nabla F(\epsilon(p_0))^T \Delta p < 0$. This means that the Newton direction is a direction of descent \cite{nocedal2006numerical}. Due to the convexity of $F,$ a local minima will correspond to global minima. 

Since optimization problems such as ours do not ensure a convex Hessian ($H$ may not be a positive definite matrix, thus leading to saddle and multiple local minima points), $H$ is approximated to produce a positive definite matrix. Let us list the necessary conditions. L-BFGS-B belongs to the Quasi-Newton category, where the Hessian matrix is updated at every iteration rather than being calculated from scratch. Let us begin with Equation \ref{eqn:Taylor_approx_of_F}. Suppose that in an iteration $i$, the parameterization vector $p$
is indicated by $p_i$. Let $\alpha$ at this point be denoted by $\alpha_i$. $p$ is updated by the following equation.

\begin{equation*}
    p_{i+1}= p_i + \alpha_i \Delta p_i
\end{equation*}

Here, $\Delta p_i= -\frac{\nabla F(\epsilon(p_i))}{H(p_i)}$, which can be written as $-H(p_0)^{-1}\nabla F(\epsilon(p_0))$. Suppose that we wish to estimate the $H$ matrix for $p_{i+1}$ \cite{nocedal2006numerical}. It will be used to construct the following. 

\begin{equation*}
\begin{split}
     F(p_{i+2})= F(p_{i+1}) + \nabla F(p_{i+1})^T (p_{i+2}-p_{i+1}) \\ 
     + \frac{1}{2}(p_{i+2}-p_{i+1})^T H(p_{i+1}) (p_{i+2}-p_{i+1})
\end{split}
\end{equation*}

Let us perform $p_i - p_{i+1}= - \alpha_i \Delta p_i$. 

\begin{equation}
    \begin{split}
        \nabla F(p_{i})= \nabla F(p_{i+1}) + H(p_{i+1}) (-\alpha_i \Delta p_i) \\\implies H(p_{i+1}) \alpha_i \Delta p_i = \nabla F(p_{i+1}) - \nabla F(p_i)
    \end{split}
\end{equation}
 
Thus, Hessian can be computed at every step just from the gradient of the objective function $F$ subject to the following constraint. 
Let us call ``$\alpha_i \Delta p_i$'' $v_i$ and ``$(\nabla F(p_{i+1}) - \nabla F(p_{i}))$'' $y_i$. 

\begin{equation}
    H(p_{i+1})v_i = y_i
\end{equation}

To ensure $H$ is a positive definite matrix, $v_i$ and $y_i$ must satisfy the following.

\begin{equation}
    v_i ^T H(p_{i+1}) v_i = v_i ^T y_i > 0 
    \label{eqn:curvature_condition}
\end{equation}

The search of $\alpha$ is influenced so that the constraint of \ref{eqn:curvature_condition} is satisfied. The Wolfe
conditions state that $\alpha$ in every iteration must ensure that a significant decrease is observed in the objective function. 
The decrease can be quantified by the following inequality in Equation \ref{eqn:significant_decrease_in_F_condition} \cite{linesearch}. Let us define a constant $c_1 \in (0,1)$.

\begin{equation}
\begin{split}
      F(p_{i+1}) \leq F(p_i) + c_1 \nabla F(p_{i})^T\alpha_i \Delta p_i 
\end{split}
\label{eqn:significant_decrease_in_F_condition}
\end{equation}

$c_1$ is kept as $10^{-4}$ \cite{linesearch}. The second condition (\ref{eqn:termination_condition_for_line_search}) terminates
the search in the direction $\alpha_i\Delta p_i$ if it does not produce any further decrease in $F$. 

\begin{equation}
    \nabla F(p_{i+1}) ^T \alpha_i \Delta p_{i} \ge c_2 \nabla F(p_i)^T \alpha_i \Delta p_i
    \label{eqn:termination_condition_for_line_search}
\end{equation}
$c_2$ is assigned $0.9$ \cite{linesearch}. It has been proved that an $\alpha$ will always exist that satisfies the Wolfe conditions. 
Suppose $F$ is convex. The constraint of Equation \ref{eqn:curvature_condition} is inherently satisfied by $F$. 
The reason is as follows. The first order of convexity condition of convexity is the following \cite{wang2021convexity}.

\begin{equation}
    F(y) \ge F(x) + \nabla F(x)^T (y-x) \forall x,y \in \text{domain of $F$}
    \label{eqn:first_order_convexity_rule_of_F}
\end{equation}

Using Equation \ref{eqn:first_order_convexity_rule_of_F}, we can show the following.

\begin{equation}
    F(p_{i+1}) \ge F(p_i) + \nabla F(p_i)^T (p_{i+1}-p_i)
    \label{eqn:convex_rl1}
\end{equation}

\begin{equation}
    F(p_{i}) \ge F(p_{i+1}) + \nabla F(p_{i+1})^T (p_{i}-p_{i+1})
    \label{eqn:convex_rl2}
\end{equation} 

 Adding Equations \ref{eqn:convex_rl1} and \ref{eqn:convex_rl2} results in $(p_{i+1} - p_i )(\nabla F(p_{i+1})- \nabla F(p_i)) \ge 0$, which satisfies Equation \ref{eqn:curvature_condition}. In view of the possibility that $F$ is non-convex, enforcing Wolfe conditions (Equation \ref{eqn:termination_condition_for_line_search}) on the line search will ensure that Equation \ref{eqn:curvature_condition} is met (see Equation \ref{eqn:enforce_wolfe}). We shall use the following results established earlier: $c_2 < 1$; when $F$ is convex at the point $p_i$, $\nabla F(p_i)^T \Delta p_i < 0$.

 \begin{equation}
     \begin{split}
         \nabla F(p_{i+1}) ^T \alpha_i \Delta p_{i} & \ge c_2 \nabla F(p_i)^T \alpha_i \Delta p_i \\
         \implies (\nabla F(p_{i+1}) -\nabla F(p_i))^T (p_{i+1}-p_i) & \ge \alpha_i(c_2-1) \nabla F (p_i) ^T \Delta p_i \\ & > 0
     \end{split}
     \label{eqn:enforce_wolfe}
 \end{equation}

\begin{algorithm}[!htb]

    \caption{ BFGS Algorithm } \label{alg:lbfgs-b}
    \begin{algorithmic}
    \Require parametrisation vector $p$, objective function $F$
    \Ensure \text{Minimize $F(\epsilon(p))$ such that $0 \leq p \leq 1$}
    \State \text{Step 1: Choose the starting vector $p_0$ }
    \State $p_0 \gets \text{vector of random values}$
   \State $B_0 \gets I$ \Comment{Initial inverse Hessian assigned identity matrix}
    \State $i \gets 0$ \Comment{Iteration number}
    
    \While{\text{Convergence reached}}

    \State \text{Step 2: Compute the gradient of $F$ at $p_i$}
    \State \text{Step 3: Perform Backtracking Line Search \cite{linesearch}  to find $\alpha_i$} \Comment{Fulfils Wolfe conditions; uses only gradient information}
    
    \State \text{Step 4: Compute search direction and Hessian at $p_{i+1}$ }   \Comment{to perform updates on $p$}  
    
    \If{$i = 0$} 

    \State $\Delta p_i \gets - B(p_i)\nabla F(p_i)$ \Comment{$B=$ inverse Hessian}
    \State $p_{i+1} \gets p_i + \alpha_i \Delta p_i$
    \State $v_i \gets p_{i+1} - p_i$
    \State $y_i \gets \nabla F(p_{i+1}) - \nabla F(p_i)$
    
    $B_i \gets \frac{y_i^T v_i}{y_i^T y_i} I$ \Comment{Reset $B_0$ for further updates \cite{quasi_newton}}
    
    \EndIf

     \State $\Delta p_i \gets - B(p_i)\nabla F(p_i)$ \Comment{$B=$ inverse Hessian}
    \State $p_{i+1} \gets p_i + \alpha_i \Delta p_i$
    \State $v_i \gets p_{i+1} - p_i$
    \State $y_i \gets \nabla F(p_{i+1}) - \nabla F(p_i)$
    \State $\rho_i \gets \frac{1}{y_i^Tv_i}$
    \State $B_{i+1} \gets (I- \rho_i v_i y_i^T)B_i (I- \rho_i y_i v_i^T) + \rho_i v_i v_i^T$ \Comment{Updating inverse Hessian as per BFGS formula \cite{quasi_newton}}

    \State $i \gets i+1$
    
    \EndWhile
    \end{algorithmic}
\end{algorithm}

\section{Symmetry Injection: Monte Carlo Search Tree} 
\label{sec:search_tree}
\subsection{Injecting Periodicity}

Let us say that the z direction sees translation symmetry. Let us define a translational operator $\hat{T_d}$ such that $\hat T_d \epsilon(r)=\epsilon(r-d)=\epsilon(r)$. It can be observed that a mode with a function of the form $e^{ikz}$ is an eigenfunction of any translational operator in the z direction. 

\begin{equation*}
    \hat T_{d\hat z}e^{ikz}=e^{ik(z-d)}=(e^{-ikd})e^{ikz}
\end{equation*}
  
 If the system experiences translational symmetry in 3 directions, the modes of EM waves are of the form $\vec{H}_k(\vec{r})=\vec{H_0}(\vec{r})e^{i\vec{k}.\vec{r}}$, where $\vec{H_k}(\vec{r})$ is the electromagnetic mode, $\vec{H_0}(\vec{r})$ is a periodic function such that $\vec{H_0}(\vec{r})= \vec{H_0}(\vec{r} + \vec{R}) (\vec{R}$ is the vector that defines where the structure becomes periodic) and $\vec{k}$ is the wave vector. This is known as Bloch's theorem \cite{photonic_crystal}. Suppose that the system has a continuous translation symmetry in the $x$, $y$ directions, the in-plane vector $\vec{k}$ becomes $k_x\hat x+k_y\hat y$. Bloch rotations and mirror shifting in the dielectric distribution can introduce Bloch phase shifts in the path of light, giving direct control over the fields. We propose a symmetry introduction method in the distribution of dielectric media in the following way.  

 Let us treat the cavity matrix $\epsilon$ as a 2D matrix $M$ with dimensions $n,m$ since we treat the third dimension $z$ as free. Let us say axes 0 and 1 represent the $x$ and $y$ axes, respectively. We shall be perturbing the matrix along these two axes. Let us say that we mirror shift the cavity matrix by 1 unit along the $x$ axis. The process replaces $M_{i,j}$ with $M_{i+1,j} \forall 0 \le i \le n-2$. In row $i= n-1$, it retains the elements originally present in row $i= n-2$. Here, $n \ge 2$ is expected. Similar logic follows for mirror shift along axis 1 (or the $y$ axis). Suppose that we Bloch shift by 1 unit along axis 0 and rotate the matrix by $\frac{\pi}{2}$, the process replaces $M_{i,j}$ with $M_{i+1,j} \forall 0 \le i \le n-2$. In row $i= n-1$, it places the elements in row $i=0$ after multiplying the elements by $e^{-i\frac{\pi}{2}}$. For the purpose of illustration, we present the following examples. 
 
 \begin{enumerate}
     \item Mirror shift: Suppose that there is a matrix $M=$
    $\begin{pmatrix}
         1 & 2 \\
         2 & 1 \\
         3 & 7
     \end{pmatrix}$. Mirror shifting $M$ by +1 along axis 0 would mean transforming $M$ to $M'$ where $M'$ is  $\begin{pmatrix}
         2 & 1 \\
         3 & 7 \\
         3 & 7
     \end{pmatrix}$.
     
     \item Bloch shift and rotate: Bloch shifting $M$ by +1 along the axis 0 with Bloch phase $\frac{\pi}{2}$ would mean transforming $M$ to $M'$ where $M'$ is $\begin{pmatrix}
         2 & 1 \\
         3 & 7 \\
         -2i & -1i
     \end{pmatrix}$.
 \end{enumerate}

\subsection{Search Tree Construction}

Let us construct a search tree to navigate the cavity indices where symmetry must be introduced. Let us say that the cavity matrix $M$ has dimensions $(N_x, N_y, 1)$. Here, the cavity matrix is initialized by $\epsilon(p)$, where $p$ is produced by the function $X(\vec{\theta})$. Let the root node begin with a matrix of 1 by 1 in the middle of the cavity, which is $M [\left \lfloor{\frac{N_x}{2}} \right \rfloor -1: \left \lfloor{\frac{N_x}{2}} \right \rfloor, \left \lfloor{\frac{N_y}{2}} \right \rfloor -1: \left \lfloor{\frac{N_y}{2}} \right \rfloor, 1 ]$. The search tree gets constructed from here. At every node, the expansion space includes the following operations on the stored cavity indices.

\begin{enumerate}
    \item Shift by -1, 0, or 1 to left 
    \item Shift by -1, 0, or 1 to right
    \item Shift by -1, 0, or 1 up
    \item Shift by -1, 0, or 1 down
    \item Apply Bloch phase along axis 0 ($x$-axis) and shift either by 1 or -1
    \item Apply Bloch phase along axis 1 ($y$-axis) and shift either by 1 or -1
    \item Mirror shift along axis 0 and shift either by 1 or -1
    \item Mirror shift along axis 1 and shift either by 1 or -1
\end{enumerate}

Bloch phases are determined by the following. 
\begin{enumerate}
    \item If axis 0 is picked, the Bloch phase becomes the vector [$\frac{2\pi}{\lambda}l_x$,0,0].
    \item If axis 1 is picked, the Bloch phase becomes the vector [0,$\frac{2\pi}{\lambda}l_y$,0], where $l_x$ is the length of the selected cavity region along axis 0 and $l_y$ is the length of the selected cavity region along axis 1. 
\end{enumerate}

The forbidden moves encompass indices that fall beyond the bounds of the matrix dimensions and the situation where no cells are shifted at all (shift by 0 in all four directions). These two constraints halt the tree expansion to arrive at the leaf nodes. The number of possible moves at every step is $(3^4-1)*(2^3+1)= 720$. 

\subsection{Tree Traversal}

We use a Monte Carlo Tree Search-based approach to arrive at an optimal selection of moves. Monte Carlo trees have four components \cite{swiechowski2022monte}, which are described in the following. 

\begin{enumerate}
    \item Selection: The best move (least costly move) is selected based on a selection policy, which is the maximization of the Upper Confidence Bound (Equation \ref{eqn:ucb}), if the node has already explored all the possible moves in earlier iterations; otherwise, move to the expansion step.
    \item Expansion: This step expands the node with a randomly chosen move out of all the permitted moves and moves to the simulation step.
    \item Simulation: This step continues to expand the search space using randomly chosen legal moves, evaluates the objective function, and ends the simulation by reaching the terminal state, proceeding to the backpropagation step.
    \item Backpropagation: This step involves traversing backward to reach the root node and updating the visited nodes with their cost and the number of visits along the way. 
\end{enumerate}

Each visited node is assigned a dielectric distribution matrix based on the perturbations performed, which is used to calculate the objective function of Equation \ref{eqn:obj_function}. The value of the objective function serves as the cost of the node. The goal of the tree traversal is to eventually arrive at the ``least costly'' node. In the selection phase, the best move is picked on the basis of the ``Upper Confidence Bound'', which is defined as follows.

\begin{equation}
    -\frac{\frac{1}{1 + \exp(-\text{node.value})}}{\text{node.visits}} + \text{exploration\_weight} \cdot \sqrt{\frac{\log(\text{parent.visits})}{\text{node.visits}}}
    \label{eqn:ucb}
\end{equation}

Here, node.value carries the accumulated costs (by backpropagation) that result from the evaluation of the objective function $F(\epsilon)$ for the perturbed $\epsilon$ stored by the node; node.visits carries the number of times the node has been visited, while parent.visits stores the visits made to the parent of the current node in previous iterations. The simulation reaches the terminal state when it consumes the maximum number of iterations, or the change in the objective function drops below a threshold, or there are no further moves to make (leaf node reached). Once the terminal state is reached, the cost and the number of visits made to the node are back-propagated. Here, the exploration weight is kept as $\sqrt{2}$ \cite{swiechowski2022monte}. 

\begin{algorithm}[!htb]
\caption{Monte‑Carlo Tree Search for Objective‑Function Minimisation through Symmetry Injection}
\label{alg:mcts}
\begin{algorithmic}[1]
\Require Root node $n_0$; iteration budget $N$; exploration weight $c$ (default $\sqrt{2}$)
\Ensure  Child of $n_0$ that yields the minimum objective value

\Function{MCTS}{$n_0,\,N,\,c$}
    \State $n_0.\mathrm{root} \gets n_0$
    \For{$i \gets 1$ \textbf{to} $N$}                       \Comment{\textbf{Main loop}}
        \State $n \gets \Call{Selection}{n_0,c}$               \Comment{Step 1}
        \State $n \gets \Call{Expansion}{n}$                   \Comment{Step 2}
        \State $z \gets \Call{Simulation}{n}$                  \Comment{Step 3}
        \State \Call{Backpropagation}{$z$}                     \Comment{Step 4}
    \EndFor
    \State \Return \Call{TraceBestMove}{$n_0$}
\EndFunction
\end{algorithmic}
\end{algorithm}


\begin{algorithm}[!htb]
\caption{Node‑level function calls}
\label{alg:node_ops}
\begin{algorithmic}[1]
\Require Node $n$; exploration weight $c$; threshold $\varepsilon$ (default $10^{-5})$
\Function{Selection}{$n,\,c$}
    \While{$n$ is fully expanded \textbf{and} $n$ has children}
        \State $n \gets \Call{BestChild}{n,c}$
    \EndWhile
    \State \Return $n$
\EndFunction

\Function{BestChild}{$n,\,c$}
    \Comment{Tweaked UCB1 – minimize objective}
    \ForAll{child $n_i \in n.\mathrm{children}$}
        \State $\textstyle
            w_i \gets
            -\dfrac{\sigma(n_i.\mathrm{value})}{n_i.\mathrm{visits}}
            \;+\;
            c\sqrt{\dfrac{\ln n.\mathrm{visits}}{n_i.\mathrm{visits}}}$
    \EndFor
    \State \Return child with \textbf{max} $w_i$
\EndFunction

\Function{Expansion}{$n$}
    \If{$n$ is \textbf{not} fully expanded}
        \State Randomly pick $a\in$ legal actions of $n.\mathrm{state}$
        \State $s' \gets n.\mathrm{state}.\Call{take\_action}{a}$
        \State Create child $n'$ with state $s'$, parent $n$, move $a$
        \State Append $n'$ to $n.\mathrm{children}$
        \State \Return $n'$
    \Else
        \State \Return $n$
    \EndIf
\EndFunction

\Function{Simulation}{$n$}   \Comment{Random playout with convergence checks}
    \State $s \gets n.\mathrm{state}$;\; $v\_0 \gets s.\Call{evaluate}{}$;\; $t\gets 0$
    \Repeat
        \If{no legal actions in $s$}\Comment{leaf}
            \State \Return $n$
        \EndIf
        \State Draw random action $a$ and update $s \gets s.\Call{take\_action}{a}$
        \State $v \gets s.\Call{evaluate}{}$,\; $t\gets t+1$
    \Until{$|v - v_0| < \varepsilon$ \textbf{or} $t = t_{\max}$}
    \State \Return new terminal node holding $s$
\EndFunction

\Procedure{Backpropagation}{$z$}
    \While{$z \neq \mathtt{null}$}
        \State $z.\mathrm{visits} \gets z.\mathrm{visits}+1$
        \State $z.\mathrm{value} \gets z.\mathrm{value} + z.\mathrm{state}.\Call{evaluate}{}$
        \State $z \gets z.\mathrm{parent}$
    \EndWhile
\EndProcedure

\Function{TraceBestMove}{$n_0$}
    \State \Return child of $n_0$ (recursively) with minimum stored value
\EndFunction
\end{algorithmic}
\end{algorithm}

\section{Results} \label{sec: results}  

For evaluation purposes, we keep $\epsilon_1$ as $1.5$ and $\epsilon_2$ as $3.45$.

Figure \ref{fig:cost_vs_parameter_theta} shows the evolution of the function $H(\vec{\theta})$, where $\theta_{11}$ and $\theta_{12}$ are independent parameters. 

\begin{figure}[!htb] 
\centering
\includegraphics[width=1\linewidth]{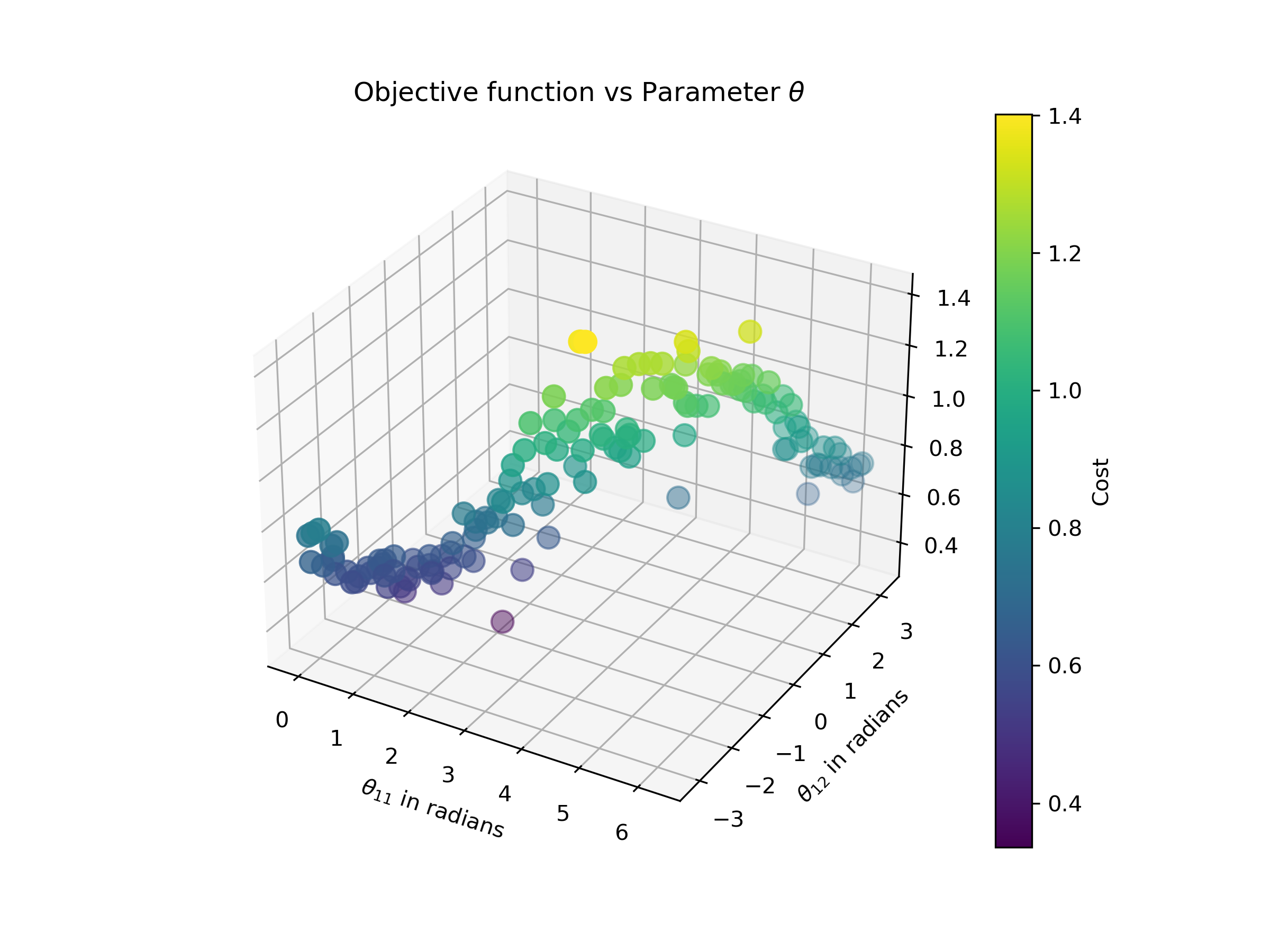}
\caption{Effect of Parameter $\theta$ on Objective Function}
\label{fig:cost_vs_parameter_theta}
\end{figure}

\begin{figure}[!htb] 
\centering
\includegraphics[width=1\linewidth]{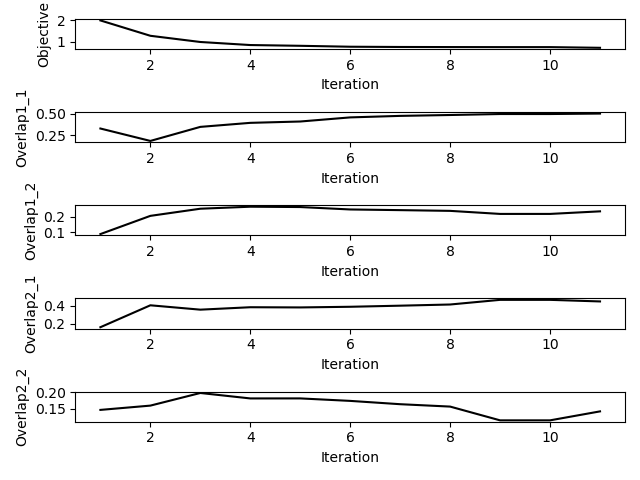}
\caption{Cost Optimisation via Gradient-Based Learning}
\label{fig:grad_descent_optimisation}
\end{figure}

Figure \ref{fig:grad_descent_optimisation} shows the gradient-supported cost optimization $X(\vec{\theta}),$ where $\vec{\theta}$ is $[2.19, -1.14, 2.19, 1.99],$ which is $\vec{\theta}$ where $H(\vec{\theta})$ achieves its minima (Figure \ref{fig:cost_vs_parameter_theta}). The results begin to converge within 20 iterations. In the figure, the terms ``Overlap$i\_j$'' refer to the term $||C_{ij}||^2$ in Equation \ref{eqn:F_in_terms_of_overlap_vec}.

\begin{figure}[!htb]
\centering
\includegraphics[width=1\linewidth]{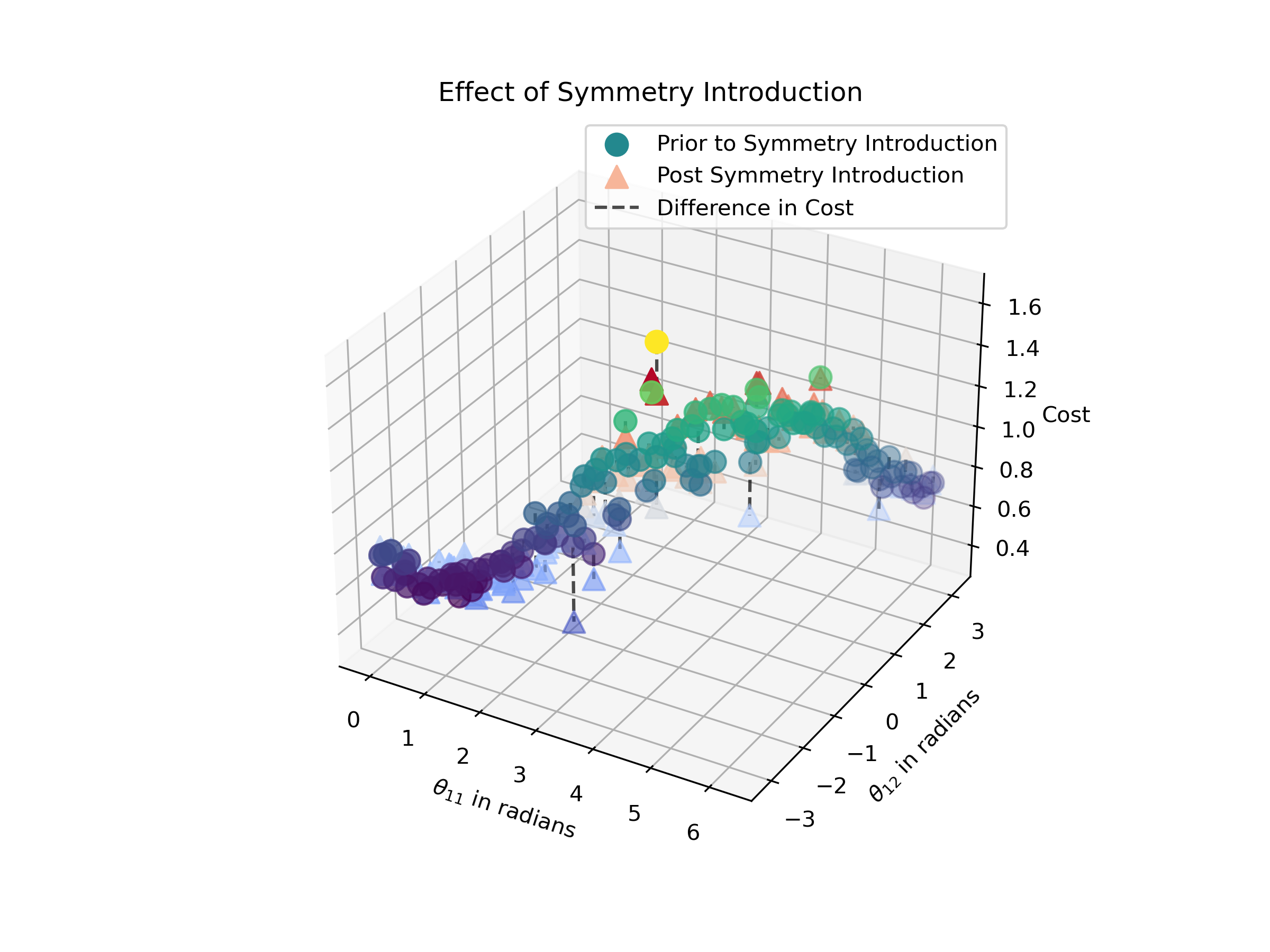}
\caption{Effect of Symmetry Introduction vs Parameter $\theta$}
\label{fig:symmetry_intro}
\end{figure}

Figure \ref{fig:symmetry_intro} shows the degree of improvement introduced into the system by introducing symmetry into the cavity. At $\vec{\theta}$ $=[2.19, -1.14, 2.19, 1.99]$, the cost is reduced to 0.335 from the initial cost of 0.714 achieved by gradient-based learning. Figures \ref{fig:unflipped} and \ref{fig:flipped_fields} show light propagation in the cavity before and after injection of symmetry. 

\begin{figure}[!htb]
\centering
\includegraphics[width=1\linewidth]{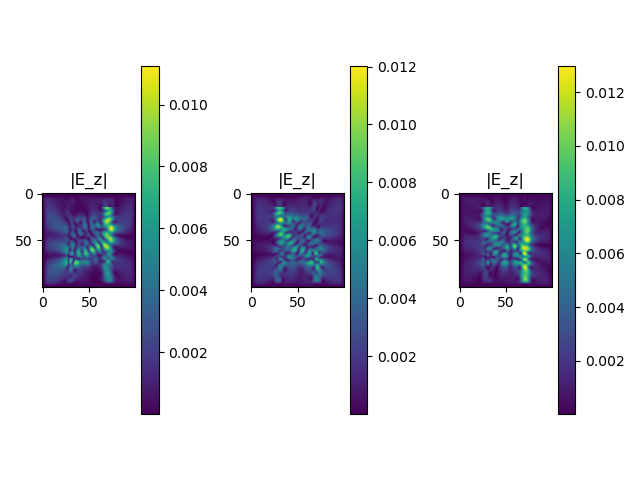}
\caption{Fields Before Symmetry Introduction }
\label{fig:unflipped}
\end{figure}

\begin{figure}[!htb]
\centering
\includegraphics[width=1\linewidth]{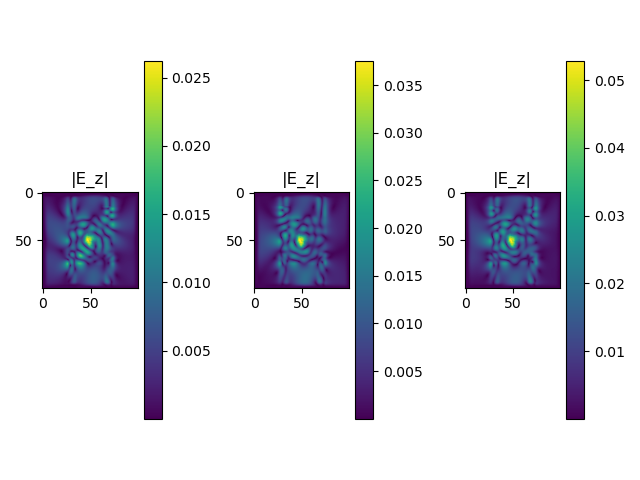}
\caption{Fields Corresponding To Cavity Post Symmetry Introduction }
\label{fig:flipped_fields}
\end{figure}

Let us consider $x$ and $y$ as the two given numbers ($\in [-1,1]$) whose product will be calculated ($=xy$). Taking into account $J$ as means to encode data, the dielectric distributions in the cavity were optimized keeping the two sources $J_1 $ and $J_2$ as unit impulse signals (putting it in terms of $x$ and $y:$ $x=1, y=1$). 
The cavity response to the unit impulse signals is taken as the output of the cavity. The most optimized version of the optical cavity showed the following response.

\begin{table}[!htb]
\centering
\begin{tabular}{lllll}
\toprule
\textbf{Name} & \textbf{Value} & \textbf{Amplitude} & \textbf{Phase}(in radians) \\
\midrule
$C_{11}$ & $-0.2055 + 0.5280i$ & 0.5666 & 1.9419 \\
$C_{12}$ & $0.1427 - 0.2978i$ & 0.3302 & -1.1238 \\
$C_{21}$ & $-0.3934 + 0.5078i$ & 0.6423 & 2.2299 \\
$C_{22}$ & $-0.0356 + 0.3636i$ & 0.3653 & 1.6684 \\
\bottomrule
\end{tabular}
\caption{Optical cavity response to impulse signal}
\label{tab:cavity_response}
\end{table}
Equation \ref{eqn:TE_mode_excitation_coefficient} is used as the recipe for encoding numbers $x,y$ in terms of the amplitude of the source current densities and the direction of propagation (forward/backward propagation). From Table \ref{tab:cavity_response}, the cavity response can be generalized in terms of $x,y$.

The fields in the right output port can be generalized as
\begin{equation}
    (-0.2055 + 0.5280i)\,x + (-0.3934 + 0.5078i)\,y
    \label{eqn:general_op1}
\end{equation}

The fields in the left output can be generalized as follows.

\begin{equation}
    (0.1427 - 0.2978i)\,x + (-0.0356 + 0.3636i)\,y
    \label{eqn:general_op2}
\end{equation}

Using Equations \ref{eqn:general_op1} and \ref{eqn:general_op2}, the current response through the photodetectors placed in the output ports can be generalized as follows.  

\begin{equation}
    \begin{split}
        \quad 
& \left\lVert (-0.2055 + 0.5280i)\,x + (-0.3934 + 0.5078i)\,y \right\rVert^2 \\
& \quad - \left\lVert (0.1427 - 0.2978i)\,x + (-0.0356 + 0.3636i)\,y \right\rVert^2 \\
= \; & 0.19\,x^2 + 0.16\,y^2 + 0.8\,xy \;\propto\; xy
    \end{split}
    \label{eqn:predicted_output_current}
\end{equation}

\begin{figure}[!htb]
\centering
\includegraphics[width=1\linewidth]{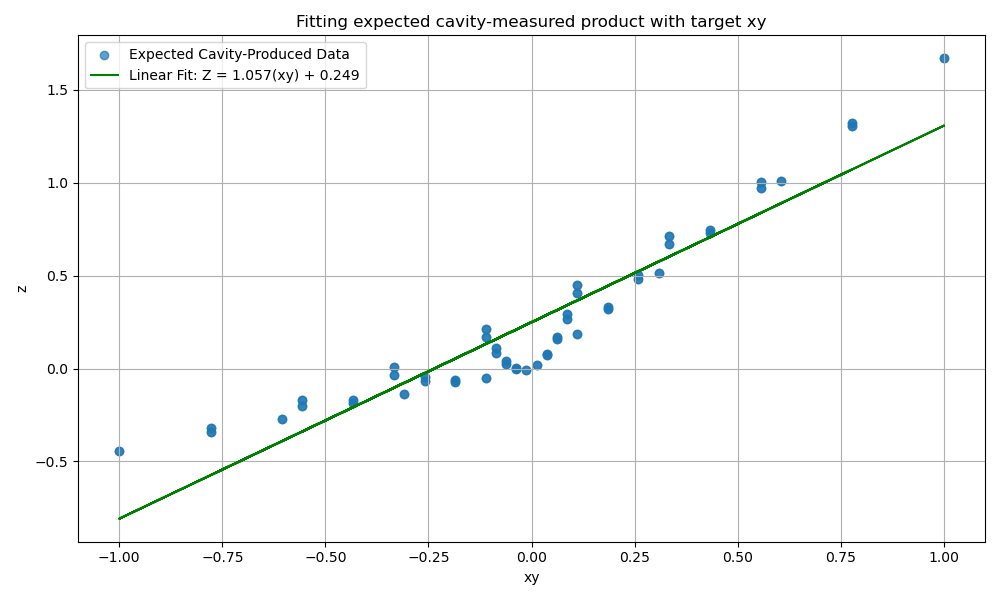}
\caption{Measuring the constant of proportionality between the expected output current produced by cavity and the target product $xy$}
\label{fig:proportionality_constant}
\end{figure}

In Figure \ref{fig:proportionality_constant},  ``z'' denotes the expected output current, generated by the optical cavity, that is measured by the photodetectors in Equation \ref{eqn:predicted_output_current}.  If we were to estimate how well the expected cavity-produced currents align with the target $xy$, the ordinary least squares method finds the relation $z= 1.057 xy + 0.249$ 
with $R^2=0.88.$ The $R^2,$ known as the coefficient of determination, lies in $[0,1]$. Since the calculated $R^2$ is close to 1, 
the estimated relation between $z$ and $xy$ can be called a good fit. Thus, the expected cavity response can be said to be directly proportional to $xy$ (as previously predicted by Equation \ref{eqn:predicted_output_current}).  

\begin{figure}[!htb]
\centering
\includegraphics[width=1\linewidth]{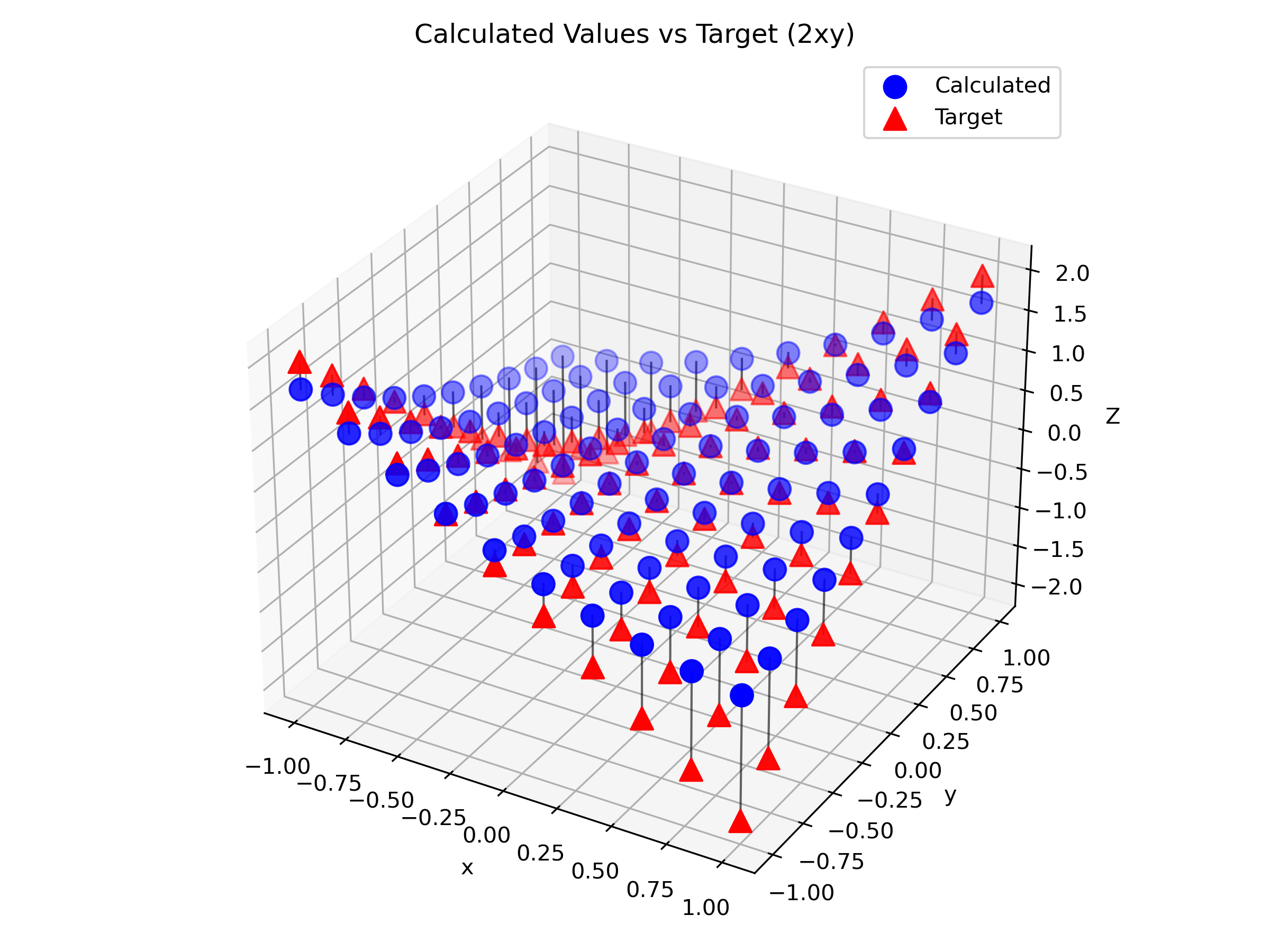}
\caption{Measuring the contrast between the target ($2xy$) and the expected cavity-produced product of $x,y$}
\label{fig:target_vs_predicted_xy}
\end{figure}

The target value is $2xy$. Figure \ref{fig:target_vs_predicted_xy} shows the comparison between the expected output current produced by the cavity (from equation \ref{eqn:predicted_output_current}) and the target value. In addition, Figure~\ref{fig:target_vs_measured_xy} shows the differences between the expected output of the cavity and the actual output. In the figure, the target is the expected value by solving equation \ref{eqn:predicted_output_current} for $x,y$.
The cavity response achieves perfect alignment with the expected response. 

\begin{figure}[!htb]
\centering
\includegraphics[width=1\linewidth]{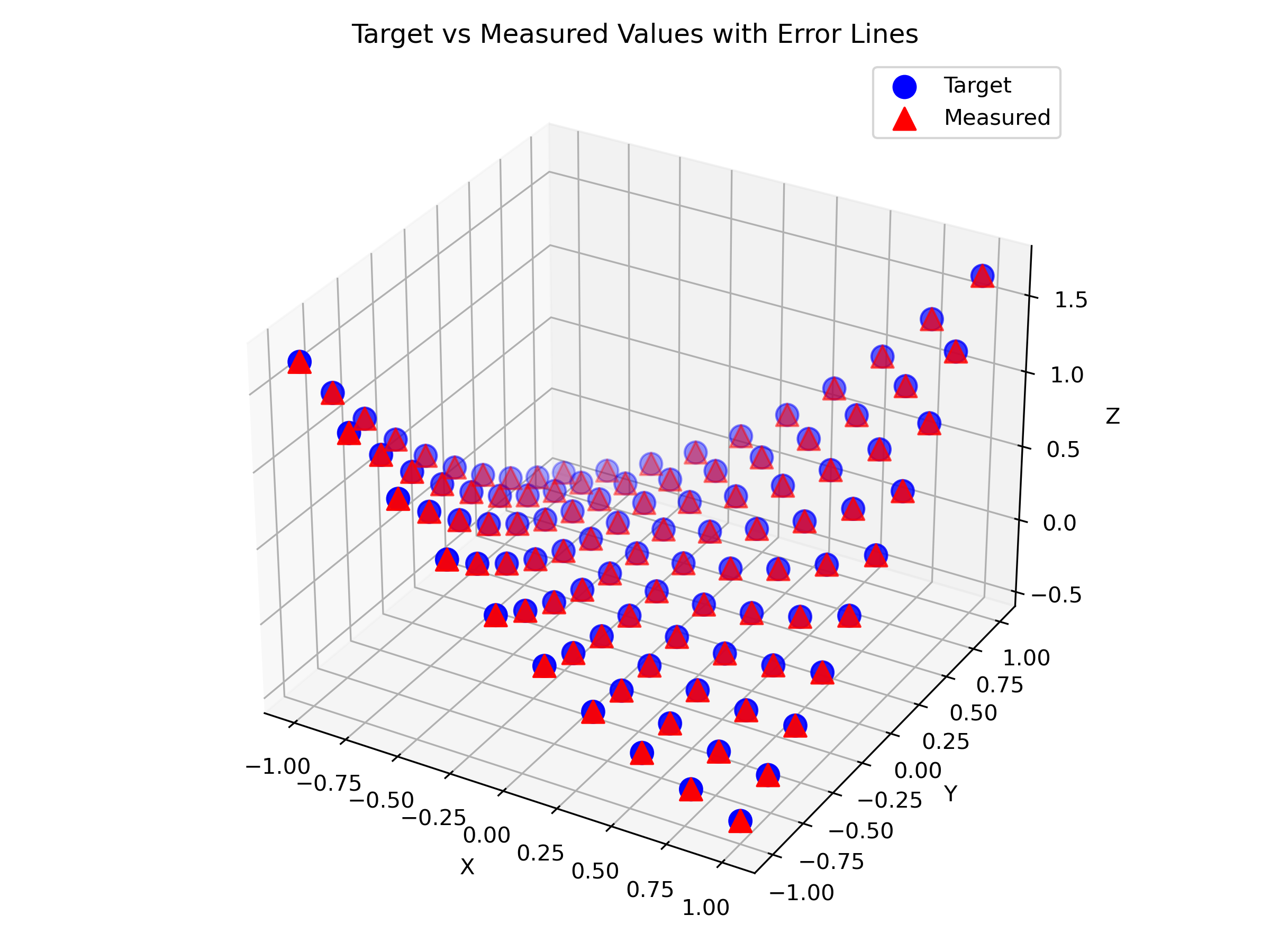}
\caption{Measuring the output current with changing source current densities}
\label{fig:target_vs_measured_xy}
\end{figure}

The \href{https://github.com/zhuhanqing/Lightening-Transformer}{LT} simulator \cite{zhu2023lighteningtransformerdynamicallyoperatedopticallyinterconnectedphotonic} 
is used to obtain simulation results of photonic accelerators. In Tables \ref{tab:area_sota_vs_our_model}, \ref{tab:energy_sota_vs_our_model} and \ref{tab:power_sota_vs_our_model}, 
we show the gains that can be achieved in terms of area, power, and energy if the photonic core of the latest state of the art ``Lightening Transformer'' were to be replaced by the proposed optical cavity, which has a footprint of $2 \times 0.22 \mu m^2$. The optical cavity replaces phase shifters, directional couplers, MZIs, and the Y-branch. Tables \ref{tab:area_sota_vs_our_model} and \ref{tab:power_sota_vs_our_model} show the simulation results for the hardware core with $4$ tiles and $2$ clusters per tile. In Lightening Transformer, photonic core occupies 18.76\% of the total area; upon introduction of the optical cavity, the photonic core ends up occupying only 2.69\% of the total area. The total area reduces to $50.36 mm^2,$ seeing a reduction of $16.5 \%$ over previous accelerators. Lasers, which previously consumed $5.22 \%$ of the total power, accounted for $4.05 \%$ of the total power when the optical cavity-based 
core is used, thus reducing the total power consumption by $1.22 \%$. 

Table \ref{tab:energy_sota_vs_our_model} calculates the energy and latency for DeiT that uses 197 tokens and operates at a precision of 4 bits. The optical cavity-based core consumes $0.88 \%$ less energy. Latency remains unaffected. This is in part due to the assumption that standard photonic devices, such as phase shifters and beam splitters, have a $0$ response time. Thus, 
replacing them with an optical cavity does not have an effect on latency. 

\begin{table*} [!htb]
    \centering
    \caption{Reduction in area when replacing existing dot product engine with inverse-designed optical cavity of dimensions (2 $\mu m$, 0.22 $\mu m$ , 2 $\mu m$)}
    \label{tab:area_sota_vs_our_model}
 
\begin{tabular}{rrrrrrr}
\toprule
\multicolumn{1}{r|}{Component} & \multicolumn{2}{r|}{MZI-based accelerator} & \multicolumn{2}{r|}{Lightening Transformer} & \multicolumn{2}{r}{Optical cavity-based photonic core} \\
\hline
& Area($mm^2$) & Percentage & Area($mm^2$) & Percentage & Area($mm^2$) & Percentage \\
ADC & 0.2736 & 0.45 & 1.6416 & 2.72 & 1.6416 & 3.26 \\
DAC & 25.3440 & 41.84 & 15.8400 & 26.26 & 15.8400 & 31.45 \\
MZM & -- & -- & 7.5942 & 12.59 & 7.5942 & 15.08 \\
TIA & 0.0048 & 0.01 & 0.0576 & 0.10 & 0.0576 & 0.11 \\
adder & 0.0512 & 0.08 & 0.0512 & 0.08 & 0.0512 & 0.10 \\
laser & 0.4800 & 0.79 & 0.7200 & 1.19 & 0.7200 & 1.43 \\
mem & 14.9022 & 24.60 & 14.6954 & 24.36 & 14.6954 & 29.18 \\
micro\_comb& -- & -- & 8.4111 & 13.94 & 8.4111 & 16.70 \\
photonic\_core& 20.0114 & 33.04 & 11.3183 & 18.76 & \fbox{1.3548}& 2.69 \\
Total& 60.5679 & 1.00 & 60.3294 & 1.00 & \fbox{50.3659}& 1.00 \\
\bottomrule
\end{tabular}

\end{table*}

\begin{table*} [!htb]
    \centering
    \caption{Reduction in power consumption when replacing existing dot product engine with inverse-designed optical cavity of dimensions (2 $\mu m$, 0.22 $\mu m$ , 2 $\mu m$)}
    \label{tab:power_sota_vs_our_model}

\begin{tabular}{rrrrr}
\toprule
\multicolumn{1}{r|}{Component} & \multicolumn{2}{r|}{Lightening Transformer} & \multicolumn{2}{r}{Optical cavity-based photonic core} \\
\hline
& Power($mW$) & Percentage & Power($mW$) & Percentage \\
ADC & 2131.2000 & 14.45 & 2131.2000 & 14.63 \\
DAC & 3214.2857 & 21.79 & 3214.2857 & 22.06 \\
MZM & 4032.0000 & 27.33 & 4032.0000 & 27.67 \\
Photodetector & 2534.4000 & 17.18 & 2534.4000 & 17.39 \\
TIA & 1728.0000 & 11.71 & 1728.0000 & 11.86 \\
adder & 26.2415 & 0.18 & 26.2415 & 0.18 \\
laser & 770.0917 & 5.22 & \fbox{589.5796}& 4.05 \\
mem & 316.3920 & 2.14 & 316.3920 & 2.17 \\
total & 14752.6109 & 1.00 & \fbox{14572.0988}& 1.00 \\
\bottomrule
\end{tabular}

\end{table*}

\begin{table*} [!htb]
    \centering
    \small
\setlength\tabcolsep{5pt}
    \caption{Reduction in energy consumption when replacing existing dot product engine with inverse-designed optical cavity of dimensions (2 $\mu m$, 0.22 $\mu m$ , 2 $\mu m$)}
    \label{tab:energy_sota_vs_our_model}

\begin{tabular}{rrrrrrr}
\toprule
\multicolumn{1}{r|}{Component} & \multicolumn{2}{r|}{MZI-based accelerator} & \multicolumn{2}{r|}{Lightening Transformer} & \multicolumn{2}{r}{Optical cavity-based photonic core} \\
\cline{2-7}
& Energy($mJ$) & Latency($ms$) & Energy($mJ$) & Latency($ms$) & Energy($mJ$) & Latency($ms$) \\
\midrule
FFN1 & 11.7423 & 50.12 & 1.7429 & 0.08 & 1.7279 & 0.08 \\
FFN2 & 11.7400 & 50.12 & 1.7360 & 0.08 & 1.7209 & 0.08 \\
attn & -- & -- & 0.1707 & 0.01 & 0.1685 & 0.01 \\
embed & 0.2434 & 1.04 & 0.0362 & 0.00 & 0.0359 & 0.00 \\
head & 0.0154 & 1.34 & 0.0140 & 0.00 & 0.0139 & 0.00 \\
others & 0.0021 & 0.00 & 0.0021 & 0.00 & 0.0021 & 0.00 \\
proj & 2.9356 & 12.53 & 0.4357 & 0.02 & 0.4320 & 0.02 \\
qkv & 8.8067 & 37.59 & 1.3072 & 0.06 & 1.2959 & 0.06 \\
total & 44.2923 & 190.34 & 5.4449 & 0.27 & \fbox{5.3970} & 0.27 \\
\bottomrule
\end{tabular}
\end{table*}

\section{Conclusion and Future Prospects}

The work explores the effect of miniaturisation of the photonic cores by leveraging optical cavities. 
Lightening Transformer, which has outperformed the prior accelerators in terms of area, energy, and latency, 
reduced the area of the photonic core by $43.4\%$, but increased the area occupied by lasers by $50\%$.  
Furthermore, it also introduced micro comb to serve as WDMs for parallel data streams, which constitute around 
$8.4\%$ of the total area, thus LT could reduce the total area only by $0.4\%.$ Enabling dot product operations 
in their photonic core through the optical cavity significantly miniaturises the former by $88\%$; therefore, when 
compared to MZI-based photonic accelerators, a core made of optical cavity in the LT photonic architecture can achieve 
reductions in overall area by $16.84\%$. As anticipated, miniaturisation of the area also reduces power and energy 
consumption by $1.22\%$ and $0.88\%$ respectively. 

Optical cavities warrant the need to devise schemes for representing data. The work encodes data in terms of the source 
current's amplitude and direction. The encoding requires data in range [-1,1]. In theory, the cavity response, governed 
by Equation \ref{eqn:TE_mode_excitation_coefficient}, succeeds in generating photocurrents that are proportional 
to the target product of the input $x$ and $y$, which is $xy$. The simulation results simulate the actual response of the cavity 
and report a response that matches its theoretical counterpart. 
Although the study did not examine discrete cavity optimization, the findings demonstrate that cavity-based photonic cores, 
augmented by microcomb integration, represent a bright direction for scalable, low-power, high-density photonic computation.

\bibliographystyle{plain}
\bibliography{refs}

\begin{thebibliography}{10}

\bibitem{mit_ocw_newtons_method}
Dimitri~P. Bertsekas.
\newblock Lecture 3: Newton's method.
\newblock \url{https://ocw.mit.edu/courses/15-084j-nonlinear-programming-spring-2004/83159d56de04a7e7dc94d0348fa4ccda_lec3_newton_mthd.pdf}, 2004.
\newblock Accessed: 2025-04-15.

\bibitem{adjmeth}
Andrew~M. Bradley.
\newblock Pde-constrained optimization and the adjoint method, 2024.

\bibitem{camacho2021single}
Miguel Camacho, Brian Edwards, and Nader Engheta.
\newblock A single inverse-designed photonic structure that performs parallel computing.
\newblock {\em Nature Communications}, 12(1):1466, March 2021.

\bibitem{Chew2016}
Weng~Cho Chew.
\newblock Theory of microwave and optical waveguides, 2016.
\newblock Lecture notes, updated February 15, 2016.

\bibitem{estakhri2019inverse}
Nasim~Mohammadi Estakhri, Brian Edwards, and Nader Engheta.
\newblock Inverse-designed metastructures that solve equations.
\newblock {\em Science}, 363(6433):1333--1338, March 2019.

\bibitem{giles2000introduction}
Michael~B Giles and Niles~A Pierce.
\newblock An introduction to the adjoint approach to design.
\newblock {\em Flow, turbulence and combustion}, 65:393--415, 2000.

\bibitem{2021Fourier}
{Hong Kong University of Science and Technology}.
\newblock Fourier {Series}, nov 18 2021.
\newblock [Online; accessed 2025-03-18].

\bibitem{photonic_crystal}
John~D. Joannopoulos, Steven~G. Johnson, Joshua~N. Winn, and Robert~D. Meade.
\newblock {\em Photonic Crystals: Molding the Flow of Light}.
\newblock Princeton University Press, USA, 2nd edition, 2008.

\bibitem{liu2019}
Yingjie Liu, Ke~Xu, Shuai Wang, Weihong Shen, Hucheng Xie, Yujie Wang, Shumin Xiao, Yong Yao, Jiangbing Du, Zuyuan He, and Qinghai Song.
\newblock Arbitrarily routed mode-division multiplexed photonic circuits for dense integration.
\newblock {\em Nature Communications}, 10(1):3263, 2019.

\bibitem{vuck_invdes}
S.~Molesky, Z.~Lin, A.Y. Piggott, et~al.
\newblock Inverse design in nanophotonics.
\newblock {\em Nature Photon}, 12(11):659--670, 2018.

\bibitem{linesearch}
Jorge Nocedal and Stephen~J. Wright.
\newblock {\em Line Search Methods}, pages 30--65.
\newblock Springer New York, New York, NY, 2006.

\bibitem{nocedal2006numerical}
Jorge Nocedal and Stephen~J. Wright.
\newblock {\em Numerical Optimization}.
\newblock Springer Series in Operations Research and Financial Engineering. Springer New York, NY, New York, NY, 2 edition, 2006.
\newblock eBook ISBN: 978-0-387-40065-5; Softcover ISBN: 978-1-4939-3711-0; Series ISSN: 1431-8598; Series E-ISSN: 2197-1773.

\bibitem{quasi_newton}
Jorge Nocedal and Stephen~J. Wright.
\newblock {\em Quasi-Newton Methods}, pages 135--163.
\newblock Springer New York, New York, NY, 2006.

\bibitem{doi:https://doi.org/10.1002/0471213748.ch5}
Bahaa E.~A. Saleh and Malvin~Carl Teich.
\newblock {\em Electromagnetic Optics}, chapter~5, pages 157--192.
\newblock John Wiley \& Sons, Ltd, Hoboken, NJ, USA, 1991.

\bibitem{doi:https://doi.org/10.1002/0471213748.ch2}
Bahaa E.~A. Saleh and Malvin~Carl Teich.
\newblock {\em Wave Optics}, chapter~2, pages 41--79.
\newblock John Wiley \& Sons, Ltd, Hoboken, NJ, USA, 1991.

\bibitem{Sapra_2020}
Neil~V. Sapra, Ki~Youl Yang, Dries Vercruysse, Kenneth~J. Leedle, Dylan~S. Black, R.~Joel England, Logan Su, Rahul Trivedi, Yu~Miao, Olav Solgaard, Robert~L. Byer, and Jelena Vučković.
\newblock On-chip integrated laser-driven particle accelerator.
\newblock {\em Science}, 367(6473):79–83, January 2020.

\bibitem{cs229_cvxopt_convex}
Stanford CS229~Course Staff.
\newblock Convex optimization with cvxopt.
\newblock \url{https://cs229.stanford.edu/section/cs229-cvxopt.pdf}, 2013.
\newblock Accessed: 2025-04-15.

\bibitem{vuckovic_spins_arch}
Logan Su, Dries Vercruysse, Jinhie Skarda, Neil~V. Sapra, Jan~A. Petykiewicz, and Jelena Vučković.
\newblock Nanophotonic inverse design with spins: Software architecture and practical considerations.
\newblock {\em Applied Physics Reviews}, 7(1):011407, 03 2020.

\bibitem{swiechowski2022monte}
Maciej {\'S}wiechowski, Konrad Godlewski, Bartosz Sawicki, and Jacek Ma{\'n}dziuk.
\newblock Monte carlo tree search: A review of recent modifications and applications, 2022.
\newblock arXiv:2103.04931 [cs.AI].

\bibitem{2023Maxwell}
{Virginia Polytechnic Institute and State University}.
\newblock The {Maxwell}-{Faraday} {Equation}, jan 25 2023.
\newblock [Online; accessed 2025-01-30].

\bibitem{wang2021convexity}
Zhanyu Wang.
\newblock Convexity, strong convexity, and smoothness in optimization.
\newblock \url{https://www.stat.purdue.edu/~wang4094/resources/slides/2021_spring_DL_meeting_01_opt_basics.pdf}, February 2021.
\newblock Accessed: 2025-04-15.

\bibitem{lbfgs}
Ciyou Zhu, Richard~H. Byrd, Peihuang Lu, and Jorge Nocedal.
\newblock Algorithm 778: L-bfgs-b: Fortran subroutines for large-scale bound-constrained optimization.
\newblock {\em ACM Trans. Math. Softw.}, 23(4):550–560, December 1997.

\bibitem{zhu2023lighteningtransformerdynamicallyoperatedopticallyinterconnectedphotonic}
Hanqing Zhu, Jiaqi Gu, Hanrui Wang, Zixuan Jiang, Zhekai Zhang, Rongxing Tang, Chenghao Feng, Song Han, Ray~T. Chen, and David~Z. Pan.
\newblock Lightening-transformer: A dynamically-operated optically-interconnected photonic transformer accelerator, 2023.

\bibitem{MaterialBoundaries}
ETH Zürich.
\newblock Material boundaries, n.d.
\newblock Lecture notes, ETH Zürich.

\end{thebibliography}

\end{document}